\documentclass[draftcls,onecolumn,12pt,a4]{IEEEtran}
\usepackage{amsmath,amsopn,amssymb,stackrel,amsthm}
\usepackage[utf8]{inputenc}

\usepackage[table]{xcolor}

\usepackage{graphicx,xspace,color,soul}
\usepackage{epsfig,subfigure}
\usepackage{longtable,multirow}
\usepackage{array,float}
\usepackage{cite,citesort}
\usepackage{url}
\usepackage{bm}
\usepackage{optidef}
\usepackage[noend]{algpseudocode}

\usepackage{algorithm,algorithmicx,algpseudocode}

\usepackage{stfloats}
\usepackage{authblk}
\usepackage{stmaryrd}

\begin{document}
\title{Joint Allocation Strategies of Power and Spreading Factors with Imperfect Orthogonality in LoRa Networks}
\author{\small \authorblockN{Licia Amichi$^{(1)}$, Megumi Kaneko$^{(2)}$, Ellen Hidemi Fukuda$^{(3)}$, Nancy El Rachkidy$^{(4)}$, and Alexandre Guitton$^{(4)}$}

\small 
\authorblockA{$^{(1)}$ INRIA Saclay, Bâtiment Alan Turing Campus de l'École Polytechnique, 91120 Palaiseau, France }

\small 
\authorblockA{$^{(2)}$ National Institute of Informatics, 2-1-2 Hitotsubashi, Chiyoda-ku, 101-8430 Tokyo, Japan}

\small 
\authorblockA{$^{(3)}$Kyoto University, Yoshida Honmachi, 606-8501 Kyoto, Japan}

\small 
\authorblockA{$^{(4)}$ Université Clermont Auvergne, CNRS, LIMOS, F-63000 Clermont-Ferrand, France} 

\small 
Email: licia.amichi@inria.fr, megkaneko@nii.ac.jp, ellen@i.kyoto-u.ac.jp, \{nancy.el\_rachkidy, alexandre.guitton\}@uca.fr} 
\maketitle
\vspace{-2cm}
\begin{abstract}
\vspace{-.2cm}
The LoRa physical layer is one of the most promising Low Power Wide-Area Network (LPWAN) technologies for future Internet of Things (IoT) applications. It provides a flexible adaptation of coverage and data rate by allocating different Spreading Factors (SFs) and transmit powers to end-devices. We focus on improving throughput fairness while reducing energy consumption. Whereas most existing methods assume perfect SF orthogonality and 
ignore the harmful effects of inter-SF interferences, we formulate a joint SF and power allocation problem to maximize the minimum uplink throughput of end-devices, subject to co-SF and inter-SF interferences, and power constraints. This results into a mixed-integer non-linear optimization, which, for tractability, is split into two sub-problems: firstly, the SF assignment for fixed transmit powers, and secondly, the power allocation given the previously obtained assignment solution. For the first sub-problem, we propose a low-complexity many-to-one matching algorithm between SFs and end-devices. For the second one, given its intractability, we transform it using two types of constraints’ approximation: a linearized and a quadratic version. Our performance evaluation demonstrates that the proposed joint SF allocation and power optimization enables to drastically enhance various performance objectives such as throughput, fairness and power consumption, and that it outperforms baseline schemes.\footnote{Part of this paper will be presented in IEEE International Conference on Communications (ICC) 2019~\cite{lic19may}.}

\end{abstract}
\vspace{-.5cm}
\begin{IEEEkeywords}
\vspace{-.2cm}
LoRa, Spreading Factor, Resource Allocation Optimization, Matching Theory
\end{IEEEkeywords}

\IEEEpeerreviewmaketitle

\section{Introduction}
\label{sec:intro}
A wide range of applications will be enabled by the advent of Internet of Things (IoT) technology, among which smart cities, intelligent transportation systems and environmental monitoring. Given the expected proliferation of such IoT devices in the near future, providing tailored wireless communication protocols with high spectral efficiency and low power consumption is becoming more and more urgent. Indeed, many of these services will depend on the future IoT Wireless Sensor Networks (WSNs), supported by the newly developed Low-Power Wide-Area Network (LPWAN) technologies such as LoRa, SigFox or Ingenu~\cite{sigF,ingn,loraA,jet17may}. The LoRa physical layer uses the Chirp Spread Spectrum (CSS) modulation technique, where each chirp encodes $2^m$ values, for Spreading Factor (SF) $m=7$ to $12$~\cite{lorP}, and which allows multiple end-devices to use the same channel simultaneously. Based on the LoRa physical layer, LoRaWAN defines the MAC layer protocol standardized by LoRa Alliance~\cite{loraAS}. It is an increasingly used LPWAN technology, as it operates in the ISM unlicensed bands and enables a flexible adaptation of transmission rates and coverages under low energy consumption~\cite{lorP}. The LoRaWAN architecture is a star topology, where end-devices communicate with the network server through gateways over several channels based on ALOHA mechanism, with duty cycle limitations~\cite{loraA}. In LoRaWAN, smaller SFs provide higher data rates but reduced ranges, while larger SFs allow longer ranges but lower rates~\cite{jet17may}.

The main issue of LoRa-based networks such as LoRaWAN is the throughput limitation: the physical bitrate varies between 300 and 50000 bps~\cite{loraAS}. In addition, collisions are very harmful to the system performance as the LoRa gateway is unable to correctly decode simultaneous signals sent by devices using the same SF on the same channel. Such interferences will be referred to as co-SF interferences. Although SFs were widely considered to be orthogonal among themselves, some recent studies have shown that this is not the case by experimentally evaluating the effects of inter-SF interferences~\cite{Dan17sep,Gui18mar,aam18sep}. Thus, authors in~\cite{ine18nov} have analyzed the effect of imperfect SF orthogonality, through the comparison of two scenarios, perfect and imperfect SF orthogonality. Authors in~\cite{War18} also analyzed the achievable uplink LoRa throughput under imperfect SF orthogonality, and have demonstrated the harmful impact of both co-SF and inter-SF interferences on the overall throughput. More recently, ~\cite{aam18sep} also unveiled a significant drop in performance when taking into account the inter-SF interferences in high-density deployments. In~\cite{luc18oct}, the authors proposed a model for analyzing the performance of a multi-cell LoRa system considering co-SF interference, inter-SF interference, and the aggregated intra and inter-cell interferences. They also highlighted the necessity for an SF allocation scheme accounting for these interferences.

In order to improve the LoRa system performance, a number of works have proposed resource optimization methods~\cite{Zhi17dec,bre17jul,nan18}. However, most papers, so far, have assumed perfect orthogonality among SFs. In particular, the authors in~\cite{Zhi17dec} designed a channel and power allocation algorithm that maximizes the minimal rate. However, no SF allocation nor SF-dependent rates were considered, despite the strong dependency of the rate to SFs. In addition, the solution of~\cite{Zhi17dec} requires instantaneous Channel State Information (CSI) feedback, which is not adapted to LoRa networks due to their energy consumption limitations~\cite{loraAS}. In~\cite{bre17jul}, a heuristic SF-allocation is proposed in addition to a transmit power control algorithm, where end-devices with similar path losses are simply assigned to the same channel with different SFs, according to their distance to the gateway. Although the issue of inter-SF interferences was highlighted, it was ignored in their proposed solution. The authors of~\cite{nan18,nan18sep} proposed a method for decoding superposed LoRa signals using the same SF, as well as a full MAC protocol enabling collision resolution, the combination of which was shown to drastically outperform LoRaWAN jointly in terms of network throughput, delay, and energy efficiency. Finally, reference~\cite{xia19jan} extended the channel allocation method of~\cite{Zhi17dec} by investigating power allocation, and proposed an algorithm based on Markov decision process modeling.

Therefore, in this work, we jointly investigate the issues of SF and transmit power allocation optimization under both co-SF and inter-SF interferences. Unlike our preliminary work~\cite{lic19may} which only considered SF allocation under fixed transmit power, and treated the cases of co-SF and inter-SF interferences separately, we now tackle the joint SF and power allocation under a generalized co-SF and inter-SF interference modeling. We focus on the problem of maximizing the minimum achievable short-term average rate in the uplink, whereby short-term average rate is defined as the average rate over random channel fading, but given a fixed position of end-devices. 
\begin{figure}[ht]
    \centering
    \includegraphics[width=10cm]{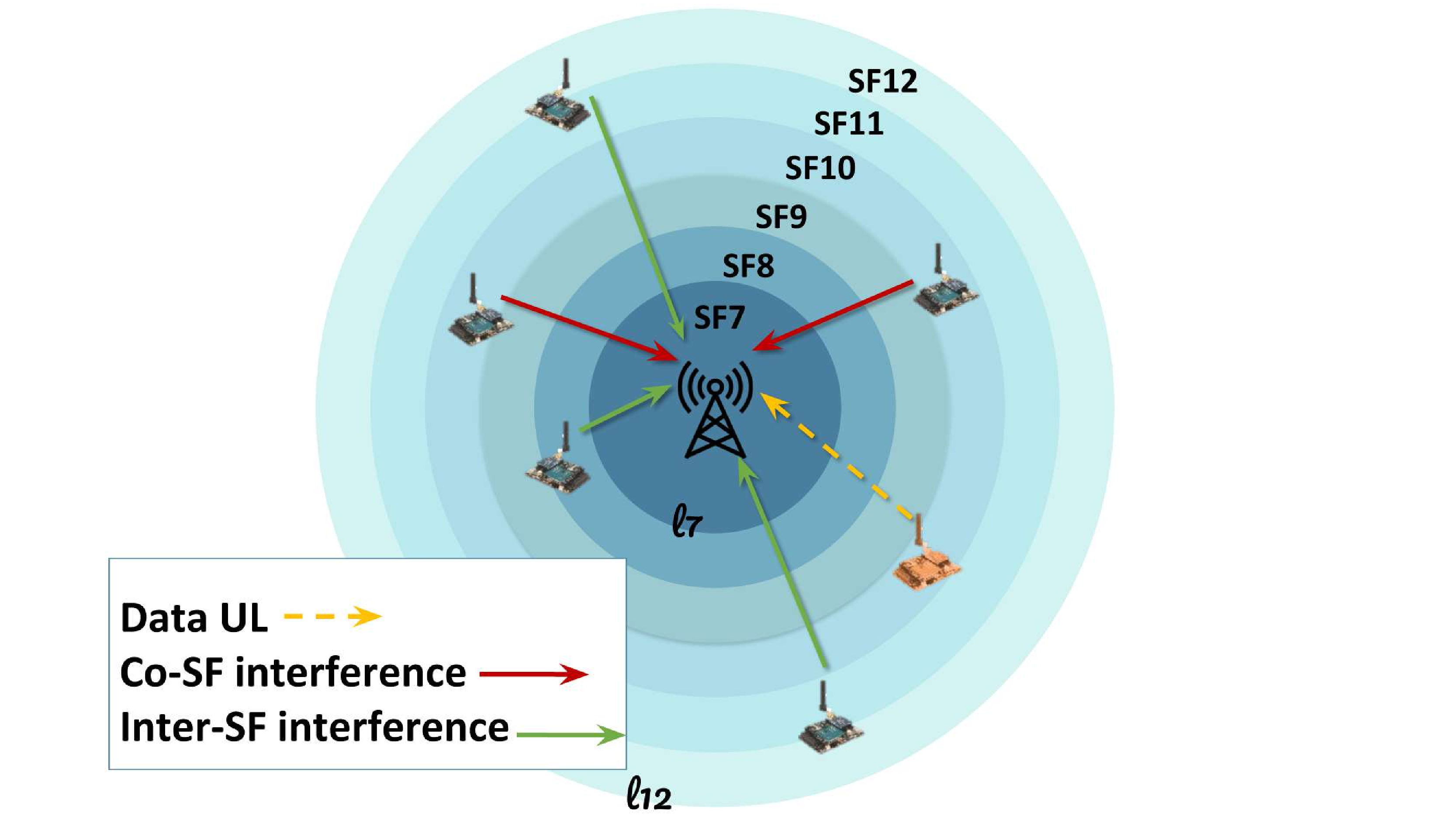}
        \vspace{-.5cm}
        \caption{LoRa network, with end-devices transmitting simultaneously on various SFs}
    \label{fig:system_setup}
    \vspace{-1cm}
\end{figure} 
This metric is especially suited for LoRa networks, since the end-devices will likely be fixed for a certain period of time (at least for a few seconds) in many applications, and their positions known at the gateway, as in conventional signal-strength-based SF allocation methods~\cite{loraAS}. Firstly, we formulate a joint SF assignment and power allocation problem by modeling the achievable uplink short-term average rate under co-SF and inter-SF interferences, and power constraints. Next, given the mathematical intractability of this mixed-integer optimization problem, we split it into two sub-problems: SF assignment under fixed transmit power, then transmit power allocation given the previous SF assignment solution. To solve the first sub-problem, we propose an SF-allocation algorithm based on matching theory. We show its stability and convergence properties, and analyze its computational complexity. Next, we transform the second sub-problem into an equivalent feasibility problem with non-linear constraints. To make it tractable, we propose to approximate the constraints in two different ways: linear and quadratic. The numerical results demonstrate that, compared to baseline schemes, our proposed method not only provides larger minimum rates, but also jointly improves the network throughput and fairness level. Moreover, the proposed power control further improves the system's performance in terms of minimum achievable rates and user fairness, while realizing massive power savings.

The remainder of this paper is organized as follows. Section~\ref{sec:system} describes the system model. Section~\ref{sec:prob} presents our joint SF and transmit power allocation problem and its contraints. Section~\ref{sec:sf-alloc} details a low-complexity many-to-one matching algorithm for the first sub-problem. Section~\ref{sec:tp-allocation} discusses our transmit power allocation scheme for the second sub-problem. Section~\ref{sec:num_res} studies the performance of the proposed algorithms. Finally, Section~\ref{sec:conclude} presents our conclusions.

\section{System Model}
\label{sec:system}
We consider a gateway located at the center of a circular cell or radius $R$ km and $N$ end-devices randomly distributed within it and simultaneously active, as depicted in Figure~\ref{fig:system_setup}. We denote by $\mathcal{N}$ the set of end-devices and by $\mathcal{M}=\{7,8,\dots,12\}$ the set of SFs. We assume that all end-devices transmit on the same channel $c$ of bandwidth $BW$, with a duty cycle of 100\% without loss of generality \footnote{LoRaWAN imposes a duty cycle of 1\% in some channels~\cite{loraAS}, in which case the theoretically achievable throughput would be 100-fold, see Section VI.}. The data bit-rate $R_m$ of $\mathrm{SF}_m, \, m\in\mathcal{M}$, is given by~\cite{lorP},
\begin{equation}\setlength{\jot}{0.1pt} R_m=\frac{m\times CR}{\frac{2^m}{BW}},\label{fig:leRm}\end{equation}
where $CR=\frac{4}{4+x}$ is the coding rate, with $x\in\{1,2,3,4\}$.

Let $h_n$ be the channel gain between the end-device $n$ and the gateway, $f_c$  the carrier frequency and $\mathrm{A}(f_c)=(f_c^2\times10^{-2.8})^{-1}$ the deterministic path-loss~\cite{War18}. Then, the uplink instantaneous Channel-to-Noise Ratio (CNR), $\zeta_{nm}$, for end-device $n$ at $\mathrm{SF}_m$ is given by~\cite{War18},
\begin{equation}
\setlength{\jot}{0.1pt}
    \zeta_{nm}=\frac{{|h_n|}^2\mathrm{A}(f_c)}{r_n^\alpha \sigma_c^2},
    \label{eq:SNR}
\end{equation}
where $r_n$ is the distance from end-device $n$ to the gateway, $\alpha$ is the path loss exponent and $\sigma_c^2=-174+\mathrm{NF}+10\mathrm{log}_{10}(BW)$ dBm is the Additive White Gaussian Noise (AWGN) and $\mathrm{NF}$ is the receiver noise figure. Assuming Rayleigh fading channels, the CNR $\zeta_{nm}$ is modeled as an exponential random variable with mean $\overline{\zeta}_{nm}=\frac{\mathrm{A}(f_c)}{r_n^\alpha \sigma_c^2}$.

The area covered by each SF is given by the distance ranges in Table~\ref{tab:tab1}~\cite{War18},
\begin{equation}
\setlength{\jot}{0.1pt}
    l_m=\mathrm{e}^{\frac{1}{\alpha}\times \mathrm{ln}\big(\frac{\mathrm{A}(f_c)}{LB_m}\big)},
\label{eq:coverage}
\end{equation}
where $LB_m$ is the link budget of the $\mathrm{SF}_m$ defined as $LB_m=P_\mathrm{{max}}-\mathrm{\theta}_{rx_m}$, given the receiver sensitivity $\mathrm{\theta}_{rx_m}$ of each $\mathrm{SF}_m$ in Table~\ref{tab:tab1} and $P_\mathrm{{max}}$ the maximal transmit power. Hence, larger SFs result in larger communication ranges, with $l_{12}=R$.

 \newcolumntype{C}[1]{>{\centering\arraybackslash}p{#1}}
\begin{table}[t]
\centering 
\resizebox{\textwidth}{!}{
\begin{tabular}{|C{1cm}|C{2cm}|C{2.5cm}|C{2.25cm}|C{2.25cm}|C{2.5cm}|}
\hline
    \textbf{SF} $m$& \textbf{Bit-rate} [kb/s] & \textbf{Receiver sensitivity~\cite{lorP}} [dBm]& \textbf{Reception thresh. }$\mathrm{\theta_{\mathrm{rx}}}_m$ [dB] & \textbf{InterSF thresh.~\cite{mar16nov} }$\mathrm{\tilde{\theta}}_{m}$ [dB] & \textbf{Distance ranges} \\
\hline
    7 &5.47&-123&-6&-7.5& [0,$l_7$] \\
\hline
    8&3.13&-126&-9&-9& ($l_7$,$l_8$] \\
\hline
    9&1.76&-129&-12&-13.5& ($l_8$,$l_9$] \\
\hline
    10&0.98&-132&-15&-15& ($l_{9}$,$l_{10}$] \\
\hline
    11&0.54&-134.5&-17.5&-18& ($l_{10}$,$l_{11}$] \\
\hline
    12&0.29&-137&-20&-22.5 &($l_{11}$,$l_{12}$] \\
\hline
\end{tabular}
}
\vspace{.1cm}
\caption{LoRa Characteristics at $BW$=125kHz \cite{War18}} 
\label{tab:tab1}
\vspace{-1.5cm}
\end{table}

Next, we denote the SF assignment by $s_{ij}$ and define it as,
\[s_{ij}=\left\{
                \begin{array}{ll}
                  1 , \text{if end-device } \,\, i \,\,\text{ uses SF } \,\, j\\
                  0, \text{otherwise.}
                \end{array}
              \right.
\]

If there is only one end-device $n$ assigned to $\mathrm{SF}_m$, this end-device is only subject to inter-SF interferences caused by end-devices using a different SF. Hence the inter-SF Signal-to-Interference-plus-Noise-Ratio (SINR) of end-device $n$ can be expressed as
\begin{equation}
\setlength{\jot}{0.1pt}
\mathrm{SINR}_{nm}^\mathrm{{inter}}=\frac{\mathrm{\zeta}_{nm}p_{nm}}{\sum\limits_{i\in\mathcal{N}_{-n}}\sum\limits_{j\in\mathcal{M}_{-m}}s_{ij}p_{ij}\zeta_{ij}+1},
\label{eq:SINR_inter}
\end{equation}
where $p_{nm}$ is the transmission power of the end-device $n$ at $\mathrm{SF}_m$, $\mathcal{N}_{-n}=\mathcal{N}\backslash\{n\}$ and $\mathcal{M}_{-m}=\mathcal{M}\backslash\{m\}$.

When there is more than one end-device assigned to a SF, these devices are subject to both inter-SF and co-SF interferences. Therefore, the co-SF SINR of device $n$ on $\mathrm{SF}_m$ is written as,

\begin{equation}
\setlength{\jot}{0.1pt}
\mathrm{SINR}^\mathrm{co}_{nm}=\frac{\mathrm{\zeta}_{nm}p_{nm}}{\sum\limits_{i\in\mathcal{N}_{-n}}\sum\limits_{j\in\mathcal{M}_{-m}}s_{ij}p_{ij}\zeta_{ij}+\sum\limits_{i\in\mathcal{N}_{-n}}s_{im}p_{im}\mathrm{\zeta}_{im}+1}.
\label{eq:SINR_co}
\end{equation}

Note that this is a more general model as compared to that of~\cite{lic19may}, which assumed the dominance of co-SF interferences over inter-SF interferences. In conformity to LoRaWAN standards, instantaneous CSI feedback is not assumed, unlike~\cite{Zhi17dec}. Hence, the SF allocation is performed every period of time, during which the long-term fading instance, i.e., path loss, can be assumed to be fixed. This is well suited to a wide range of applications envisioned for IoT systems based on LoRa, expected to be static, or with low mobility~\cite{aly16sep}. Therefore, the achievable uplink short-term average rate for end-device $n$ at $\mathrm{SF}_m$ is given similarly to~\cite{War18} by,
\begin{equation}
\tau_{nm} = R_m\times P_{\mathrm{cap}}^{(n,m)},
\label{eq:throughput}
\end{equation}
where $P_{\mathrm{cap}}^{(n,m)}$ is the probability of successful reception analyzed in the following section.

%\section{Conventional Algorithms}
%\label{sec:convA}
%\input{convA}

\section{Problem Formulation}
\label{sec:prob}
In this section, we formulate the joint SF and power allocation optimization problem in our considered LoRa-based system, under imperfect SF orthogonality. In particular, the goal will be to improve the overall fairness of the system by maximizing the minimal uplink average rate over end-devices and SFs, under co-SF and inter-SF interferences. We first derive the expression of the probability of successful reception, $P_{\mathrm{cap}}^{(n,m)}$. Assuming $N>1$, there are two cases:
\subsubsection{One end-device $n$ at $\mathrm{SF}_m$} end-device $n$ is only subject to inter-SF interferences. The transmission can be successfully decoded if the node satisfies the \textit{inter-SF} as well as the signal \textit{reception} conditions. In this case, inter-SF interferences are more critical than the signal reception condition since there are always inter-SF interferences for $N>1$. Hence the probability of successful transmission can be written as,
\begin{equation}
\setlength{\jot}{0.1pt}
\begin{split} 
P_\mathrm{{cap}_{iSF}}^{(n,m)}=P\left(\mathrm{SINR}_{nm}^\mathrm{{inter}} \geq \mathrm{\tilde{\theta}}_{m}\right),
\end{split}
\normalsize 
\label{cap_inter-SF_1}
\end{equation}
where $\mathrm{SINR}_{nm}^\mathrm{{inter}}$ is given in~(\ref{eq:SINR_inter}) and $\mathrm{\tilde{\theta}}_{m}$ is the inter-SF interference capture threshold for $\mathrm{SF}_m$, defined in Table~\ref{tab:tab1}. Using the random instantaneous CNR variables $\zeta_{nm}$ for all $(n,m)$ and marginalizing over them, it has been shown in~\cite{lic19may} with similar calculations as in~\cite{War18} that~(\ref{cap_inter-SF_1}) can be written as,
\begin{equation}
\setlength{\jot}{0.1pt}
\begin{split} 
    P_\mathrm{{cap}_{iSF}}^{(n,m)} &=\mathrm{e}^{-\frac{\mathrm{\tilde{\theta}}_{m}\sigma_c^2r_n^\alpha}{\mathrm{A}(f_c)p_{nm}}}\prod\limits_{i \in \mathcal{N}_{-n} }\prod\limits_{j \in \mathcal{M}_{-m}}\frac{1}{\mathrm{\tilde{\theta}}_{m} s_{ij}\frac{p_{ij}}{p_{nm} }\times {(\frac{r_n}{r_i})}^\alpha+1}.
\end{split}
\normalsize 
\label{cap_inter-SF_5}
\end{equation}  
\subsubsection{More than one end-device at $\mathrm{SF}_m$} in this case, the co-SF interferences as well as the inter-SF interferences largely dominate the signal reception condition~\cite{War18}. Therefore, the success probability is expressed as in~\cite{geo17apr},
\begin{equation}
\setlength{\jot}{0.1pt}
\begin{split}
    P_\mathrm{{cap}_{coSF}}^{(n,m)}&=P\left(\mathrm{SINR}^\mathrm{co}_{nm}\geq \theta_\mathrm{co}\right),
\end{split}
\label{cap_co-sf_1}
\end{equation}
\normalsize 
where $\mathrm{SINR}^\mathrm{co}_{nm}$ is given in~(\ref{eq:SINR_co}) and $\theta_\mathrm{co}$ is the co-SF capture threshold which is equal to 6dB for all $\mathrm{SF}_m$~\cite{lorP,geo17apr}.
With similar calculations as in~\cite{lic19may}, we obtain

\begin{equation}
\setlength{\jot}{0.1pt}
\begin{split}    
    P_\mathrm{{cap}_{coSF}}^{(n,m)}&= \mathrm{e}^{-\frac{\mathrm{\theta_{co}}\sigma_c^2r_n^\alpha}{ A\left(f_c\right)p_{nm}}}\left(\prod\limits_{i \in \mathcal{N}_{-n}}\prod\limits_{j \in \mathcal{M}_{-m}} \frac{1}{\mathrm{\theta_{co}}s_{ij}\frac{p_{ij}}{ p_{nm}}\times {\left(\frac{r_n}{r_i}\right)}^\alpha+1}\right)\prod\limits_{i \in \mathcal{N}_{-n}} \frac{1}{\mathrm{\theta_{co}}s_{im}\frac{p_{im}}{ p_{nm}}\times {\left(\frac{r_n}{r_i}\right)}^\alpha+1}.
\end{split}
\normalsize 
\label{cap_co-sf_3}
\end{equation} 

Given the above analysis, the joint SF and transmit power allocation optimization underlaying LoRaWAN network is formulated as follows (for $N>1$),
\begin{equation} 
\setlength{\jot}{0.1pt}
\max_{{s_{nm}}p_{nm}}\min_{\substack{(n,m)\in \\ \mathcal{N}\times \mathcal{M}\\ \text{s.t. } s_{nm}\neq 0}}f(s_{nm}, p_{nm})=s_{nm}R_m P_\mathrm{{cap}}^{(n,m)},
\label{obj_fun_0}
\normalsize 
\end{equation}
where the minimization is over the $s_{nm}$ that are non-zero, and
\begin{equation}
\begin{split}
    &P_\mathrm{{cap}}^{(n,m)}=I\left({\sum\limits_{k \in \mathcal{N}}s_{km}\geq2}\right)P_\mathrm{{cap}_{coSF}}^{(n,m)}+I\left(\sum\limits_{k \in \mathcal{N}} s_{km}=1\right)  P_\mathrm{{cap}_{iSF}}^{(n,m)},
\end{split}
\normalsize 
\end{equation}
where $I(C)$ is the indicator function, i.e., it equals 1 if the condition $C$ is verified and 0 otherwise.

Finally, the overall optimization problem becomes

\begin{subequations}\label{obj_fun_prob}
\setlength{\jot}{0.1pt}
\begin{alignat}{2} 
\nonumber&(P)\max_{{s_{nm}}p_{nm}}\min_{\substack{(n,m)\in \\  \mathcal{N}\times \mathcal{M}\\\text{s.t. } s_{nm}\neq 0}}s_{nm}R_m \Bigg[ I\left({\sum\limits_{k \in \mathcal{N}}s_{km}=1}\right) \mathrm{e}^{-\frac{\mathrm{\tilde{\theta}}_{m}\sigma_c^2r_n^\alpha}{\mathrm{A}\left(f_c\right)p_{nm}}}\prod\limits_{i \in \mathcal{N}_{-n} }\prod\limits_{j \in \mathcal{M}_{-m}}\frac{1}{\mathrm{\tilde{\theta}}_{m} s_{ij}\frac{p_{ij}}{p_{nm} } {\left(\frac{r_n}{r_i}\right)}^\alpha+1}\\&+ I\left({\sum\limits_{k \in \mathcal{N}}s_{km}\geq2}\right)\mathrm{e}^{-\frac{\mathrm{\theta_{co}}\sigma_c^2r_n^\alpha}{ A\left(f_c\right)p_{nm}}}\left(\prod\limits_{i \in \mathcal{N}_{-n}}\prod\limits_{j \in \mathcal{M}_{-m}} \frac{1}{\mathrm{\theta_{co}}s_{ij}\frac{p_{ij}}{ p_{nm}} {\left(\frac{r_n}{r_i}\right)}^\alpha+1}\right)\prod\limits_{i \in \mathcal{N}_{-n}} \frac{1}{\mathrm{\theta_{co}}s_{im}\frac{p_{im}}{ p_{nm}} {\left(\frac{r_n}{r_i}\right)}^\alpha+1} \Bigg]\tag{\ref{obj_fun_prob}}\\
&\text{s.t.       } \quad \text{  C1: } 0 \leq  p_{nm} \leq P_\mathrm{max}, p_{nm}\in\mathbb{R}^{+} \label{opt:cond_0_prob}
\\
&\quad \quad \quad \text{C2: } s_{nm} \in \{0,1\}, \text{ } \forall (n,m) \in \mathcal{N} \times \mathcal{M} \label{opt:cond_1_prob}
\\
 &\quad \quad \quad \text{C3: } \sum\limits_{m \in \mathcal{M}}s_{nm} \leq 1, \text{ }\forall n \in \mathcal{N} \label{opt:cond_2_prob}
 \\ &\quad \quad \quad  \text{C4: } \sum\limits_{n \in \mathcal{N}}s_{nm} \leq \mathrm{N_{max}}(m), \text{ }\forall m\in \mathcal{M} \label{opt:cond_3_prob}\\
  &\quad \quad \quad  \text{C5: } \text{ if } N>M, 1 \leq \sum\limits_{n\in\mathcal{N}} s_{nm}, \text{ } \forall m \in \mathcal{M} \label{opt:cond_4_prob}
\end{alignat}
\normalsize 
\end{subequations}

Our objective function~\eqref{obj_fun_prob} expresses the maximization of the minimum data-rate over all served end-devices (i.e., for which $s_{nm} \neq 0$) and SFs. Constraint~(\ref{opt:cond_0_prob}) is the power budget, where the maximum transmit power per end-device is fixed to $P_{\mathrm{max}}$. Constraint~(\ref{opt:cond_1_prob}) defines the binary SF allocation variables $s_{nm}$. Constraints~(\ref{opt:cond_2_prob}) and~(\ref{opt:cond_3_prob})\footnote{Setting $\mathrm{N_{max}}(m)$ enables to control the harmful effects of co-SF interferences, and reduces the computational complexity of the proposed method, as shown in Sections~\ref{algo_analysis} and~\ref{choice_nmax}.} ensure that an end-device $n$ is assigned to at most one $\mathrm{SF}$, and that the maximal number of end-devices sharing $\mathrm{SF}_m$ is $\mathrm{N_{max}}(m)$. Finally, ~(\ref{opt:cond_4_prob}) ensures that if there are enough end-devices ($N>M$), no SFs should remain unused, i.e., at least one end-device should be allocated to each SF. Clearly, $(P)$ is a mixed-integer problem with a non-convex objective function, as it includes both binary allocation variables $s_{nm}$ and continuous power allocation variables $p_{nm}$. Such problems are known to be generally NP-hard~\cite{rav78sep}, making them difficult to solve. We therefore propose to solve this problem by decomposing it into the following two optimization phases: 
(1) the discrete optimization phase of the allocation of binary variables $s_{nm}$ while keeping the power allocation variables $p_{nm}$ fixed to $P_\mathrm{max}$, 
(2) the continuous optimization phase of the power allocation variables $p_{nm}$, where the allocation variables have been fixed to their previous solution. 
These two phases may be iterated until convergence, or until the maximum number of iterations $N_I$ has been reached.

Denoting by $\mathbf{s}= [s_{nm}]$  and $\mathbf{p}= [p_{nm}], \forall n \in \mathcal{N}, m \in \mathcal{M}$, the SF assignment vector and transmit
power vector for all end-devices, respectively, 
Algorithm~\ref{algo:general} provides the overview of the general proposal.

\begin{algorithm}
\small
\caption{Proposed joint SF and transmit power allocation}
\label{algo:general}
\textbf{Initialization: } SF assignment vector: $\mathbf{s}^{(0)} \gets \mathbf{0}$, transmit power vector: $\mathbf{p}^{(0)}\gets P_\mathrm{max}$.
\algdef{SE}[DOWHILE]{Do}{doWhile}{\algorithmicdo}[1]{\algorithmicwhile\ #1}
\begin{algorithmic}[1]
\State $i\gets 1$.
\Do
\State SF assignment: find $\mathbf{s}^{(i)}$, for fixed $\mathbf{p}^{(i)}\gets P_{\mathrm{max}}$. \Comment{(Sec.~\ref{sec:sf-alloc})}
\State Transmit power allocation: find $\mathbf{p}^{(i)}$, for fixed $\mathbf{s}^{(i)}$.
\Comment{(Sec.~\ref{sec:tp-allocation})}
\doWhile{$f\left(\mathbf{s}^{(i)}, \mathbf{p}^{(i)}\right)-f\left(\mathbf{s}^{(i-1)}, \mathbf{p}^{(i-1)}\right) \geq \epsilon$ or $i \leq N_I$.}
\end{algorithmic}
\end{algorithm}
\vspace{-.3cm}

In the next sections, we describe each of the optimization phases.

\section{Proposed Spreading Factor Allocation}
\label{sec:sf-alloc}
\subsection{Formulation of the proposed SF allocation optimization}
In this section, the problem of SF allocation is addressed. We assume that all end-devices transmit with the maximum transmission power, i.e., $p_{nm}=P_\mathrm{{max}}, \forall n, m$. This problem can be formulated as follows,
\begin{subequations}\label{obj_fun_sf}
\setlength{\jot}{0.1pt}
\begin{alignat}{2} 
\nonumber&(P1)\max_{{s_{nm}}}\min_{\substack{(n,m)\in \\  \mathcal{N}\times \mathcal{M}\\\text{s.t. } s_{nm}\neq 0}}f(s_{nm})=s_{nm}R_m
\Bigg[I\left({\sum\limits_{k \in \mathcal{N}}s_{km}\geq2}\right)\mathrm{e}^{-\frac{\mathrm{\theta_{co}}\sigma_c^2r_n^\alpha}{ \mathrm{A}(f_c)P_{\mathrm{max}}}}  \nonumber\times\prod\limits_{i \in \mathcal{N}_{-n}} \frac{1}{\mathrm{\theta_{co}}s_{im} {(\frac{r_n}{r_i})}^\alpha+1} \\&+I\left(\sum\limits_{k \in \mathcal{N}} s_{km}=1\right) \mathrm{e}^{-\frac{\mathrm{\tilde{\theta}}_{m}\sigma_c^2r_n^\alpha}{\mathrm{A}\left(f_c\right)P_{\mathrm{max}}}}\times \prod\limits_{\substack{(i,j)\in \\ \mathcal{N}_{-n}\times \mathcal{M}_{-m}}}\frac{1}{\mathrm{\tilde{\theta}}_{m} s_{ij} {\left(\frac{r_n}{r_i}\right)}^\alpha+1}\Bigg] \tag{\ref{obj_fun_sf}}\\
&\text{s.t.   } \quad \text{   C1: } s_{nm} \in \{0,1\},\text{ } \forall (n,m) \in \mathcal{N} \times \mathcal{M} \label{opt:cond_1_sf}
\\
 &\quad\quad\quad\text{C2: } \sum\limits_{m \in \mathcal{M}}s_{nm} \leq 1, \text{ }\forall n \in \mathcal{N} \label{opt:cond_2_sf}
 \\ &\quad \quad\quad \text{C3: } \sum\limits_{n \in \mathcal{N}}s_{nm} \leq \mathrm{N_{max}}(m), \text{ }\forall m\in \mathcal{M} \label{opt:cond_3_sf}\\
  &\quad \quad\quad \text{C4: } \text{ if } N>M, 1 \leq \sum\limits_{n\in\mathcal{M}} s_{nm}, \text{ } \forall m \label{opt:cond_4_sf}
\end{alignat}
\normalsize 
\end{subequations}
$(P1)$ is an integer programming problem, given the binary variables $s_{nm}$, with a non-linear objective function, hence it is difficult to obtain its optimal solution. Therefore, we
propose an optimized SF allocation method, using tools from matching theory.\\ 

Matching theory is a promising tool for resource allocation
in wireless networks~\cite{yun15may}. According to this theory, our considered allocation problem $(P1)$ can be classified as a many-to-one matching problem with conventional externalities and peer effects. There are two sets of players, the set of SFs and the set of end-devices, where each player of the one set seeks to be matched with players of the opposing set. An end-device prefers to be matched to the SF offering the highest utility, while each SF prefers to be matched with the group of end-devices with the highest utility. The difficulty of our problem is that there is an interdependency between nodes' preferences, i.e., whenever an end-device is matched to an SF, the preferences of the other end-devices may change due to co-SF and inter-SF interferences. In addition to these conventional externalities (preference interdependency) and unlike the problem in~\cite{Zhi17dec} where only orthogonal channels (not SFs) were considered, our problem exhibits peer effects that are caused by inter-SF interferences. That is, the preferences of an end-device depend not only on the identity of the SF and the number of end-devices assigned to it, but also on the assignment of end-devices to other SFs (since they cause inter-SF interferences). Therefore, to solve $(P1)$, we propose a many-to-one matching algorithm between the set $\mathcal{M}$ of SFs and the set $\mathcal{N}$ of end-devices. Next, we define the basic concepts of matching theory.

\subsection{Fundamentals of Matching Theory}
In order to describe our proposed matching-based algorithm, we describe the basic concepts of matching theory that have been used in our algorithm:

\begin{itemize}
    \item \textbf{Matching pair:} a couple ($n$, $m$) assigned to each other.
    \item \textbf{Quotas of a player:} the maximum number of players with which it can be matched
    \begin{itemize}
        \item Each end-device has a quota of 1~(\ref{opt:cond_2_sf}),
        \item Each $\mathrm{SF}_m$ has a quota of $\mathrm{N_{max}}(m)$ end-devices~(\ref{opt:cond_3_sf}).
    \end{itemize}
    \item \textbf{Utility of an end-device:} defined for our problem as its short-term average rate. If it is the only end-device at $\mathrm{SF}_m$,
    \begin{equation}
    \setlength{\jot}{0.1pt}
    \begin{aligned}
    &U_n=R_m\mathrm{e}^{-\frac{\mathrm{\tilde{\theta}}_{m}\sigma_c^2r_n^\alpha}{\mathrm{A}\left(f_c\right)P_\mathrm{max}}}\prod\limits_{i \in \mathcal{N}_{-n} }\prod\limits_{j \in \mathcal{M}_{-m}}\frac{1}{\mathrm{\tilde{\theta}}_{m} s_{ij} {\left(\frac{r_n}{r_i}\right)}^\alpha+1}.
    \end{aligned}
    \label{eq:utilityDev_1}
    \end{equation}
    If it shares the $\mathrm{SF}_m$ with other end-devices,
    \begin{equation}
    \setlength{\jot}{0.1pt}
    \begin{aligned}
    &U_n=R_m \mathrm{e}^{-\frac{\mathrm{\theta_{co}}\sigma_c^2r_n^\alpha}{ A\left(f_c\right)P_\mathrm{max}}}\prod\limits_{i \in \mathcal{N}_{-n}}\prod\limits_{j \in \mathcal{M}_{-m}} \frac{1}{\mathrm{\theta_{co}}s_{ij}{\left(\frac{r_n}{r_i}\right)}^\alpha+1}\prod\limits_{i \in \mathcal{N}_{-n}} \frac{1}{\mathrm{\theta_{co}}s_{im}{\left(\frac{r_n}{r_i}\right)}^\alpha+1}.
    \end{aligned}
    \label{eq:utilityDev_2}
    \end{equation}
    \item \textbf{Utility of an SF:} defined for our problem as the minimum short-term average rate among the end-devices assigned to it. If $\mathrm{SF}_m$ is matched to one end-device only:
    \begin{equation}
    \setlength{\jot}{0.1pt}
    \begin{aligned}
    &U_m=R_m\mathrm{e}^{-\frac{\mathrm{\tilde{\theta}}_{m}\sigma_c^2r_n^\alpha}{\mathrm{A}\left(f_c\right)P_\mathrm{max}}}\prod\limits_{i \in \mathcal{N}_{-n} }\prod\limits_{j \in \mathcal{M}_{-m}}\frac{1}{\mathrm{\tilde{\theta}}_{m} s_{ij} {\left(\frac{r_n}{r_i}\right)}^\alpha+1},
    \end{aligned}
    \label{eq:utilitySF_1}
    \end{equation}
    otherwise $U_m$ is given as
    \begin{equation}
    \setlength{\jot}{0.1pt}
    \begin{aligned}
    &U_m=\min\limits_{n \in \mathcal{A}_m}R_m \mathrm{e}^{-\frac{\mathrm{\theta_{co}}\sigma_c^2r_n^\alpha}{ A\left(f_c\right)P_\mathrm{max}}}\prod\limits_{i \in \mathcal{N}_{-n}}\prod\limits_{j \in \mathcal{M}_{-m}} \frac{1}{\mathrm{\theta_{co}}s_{ij}{\left(\frac{r_n}{r_i}\right)}^\alpha+1}\prod\limits_{i \in \mathcal{N}_{-n}} \frac{1}{\mathrm{\theta_{co}}s_{im}{\left(\frac{r_n}{r_i}\right)}^\alpha+1}.
    \end{aligned}
    \label{eq:utilitySF_2}
    \end{equation}
    where $\mathcal{A}_m$ is the set of end-devices assigned to $\mathrm{SF}_m$.
    \item \textbf{Preference relation}: a player $q$ prefers a player $p_1$ over the player $p_2$, if the utility of $q$ is higher when it is matched to $p_1$ than when it is matched to $p_2$.
    \item \textbf{Blocking pair:} a matching pair $(n,m)$ is a blocking pair when $U_{n}$ or $U_{m}$ is higher when $n$ uses $m$, than when they use their current matches, without lowering the utilities of any other end-device nor SF. In this case, $n$ will leave its current match to be matched to $m$.
    \item \textbf{Two-sided exchange stable matching:} a matching solution where there is no blocking pair.
\end{itemize}

\subsection{Proposed SF-Allocation algorithm}
In this subsection, we describe the steps of the proposed matching-based algorithm which exploits matching techniques as in~\cite{Zhi17dec,yun15may}, tailored to our specific problem. First, the gateway performs an initial matching between the set $\mathcal{M}$ of SFs and the set $\mathcal{N}$ of end-devices by the Initial Matching in~Algorithm~\ref{algo:Initial_matching}. Next, it swaps the matching pairs obtained in the previous step until reaching a two-sided exchange stable matching by the Matching Refinement in~Algorithm~\ref{algo:Matching_refinement}. Details of these steps are given below.

Let $\mathcal{L}_{U}$ denote the set of end-devices that are not allocated to any SF, $\mathrm{req}_{m}$ the requests received by $\mathrm{SF}_m$, and $\mathcal{A}_m$ the set of end-devices assigned to $\mathrm{SF}_m$.
We suppose that the gateway knows its distance with all end-devices.\\

\textbf{Initialization:} the gateway starts by initializing the preference lists of end-devices and SFs. Each end-device $n$ with a distance $r_n$ to the gateway, can only use SFs if they are included in the coverage area ($r_n\leq l_m$) of the gateway for these SFs, therefore,
\begin{equation}
\setlength{\jot}{0.1pt}
    \mathcal{L}_{p,n}=\{m\in \mathcal{M}, \text{ s.t. } r_n \leq l_m\},
\end{equation}
$\mathcal{L}_{p,n}$ is sorted according to the increasing order of the distance threshold of the SFs ($l_m,\, m\in \mathcal{M}$), i.e., an SF with higher achievable rate is preferred. On the other hand, $\mathrm{SF}_m$ only considers end-devices having a distance to the gateway lower than $l_m$,
\begin{equation}
\setlength{\jot}{0.1pt}
    \mathcal{L}_{p,m}=\{n \in \mathcal{N}, \text{ s.t. } r_n \leq l_m\}.  
\end{equation}
$\mathcal{L}_{p,m}$ is ordered such that a user $n_1 \in \mathcal{L}_{p,m}$ is ranked before another user $n_2 \in \mathcal{L}_{p,m}$ if $n_1$ is located in the ring of $\mathrm{SF}_m$ $\left(n_{1}\in(l_{m-1},l_{m}]\right)$ but not $n_2$ $\left(n_{2}\notin(l_{m-1},l_{m}]\right)$, or both are in the ring of $\mathrm{SF}_m$ but $n_1$ is closer to the gateway than $n_2$ $\left( |r_{n_1}| < |r_{n_2}|\right)$.\\ 
%\begin{itemize}
%    \item $n_{1}\in(l_{m-1},l_{m}] \text{ and } n_{2}\notin(l_{m-1},l_{m}]$: $n_1$ is located in the ring of $\mathrm{SF}_m$ but not $n_2$,
 %  \item Or $ |l_{m-1}-r_{n_1}| < |l_{m-1}-r_{n_2}|$: both are in the ring of $\mathrm{SF}_m$ but $n_1$ is closer to the gateway than $n_2$. 
%\end{itemize}
Unmatched end-devices are added to $\mathcal{L}_{U}$.\\

\textbf{Initial Matching:} for each end-device $n$ in the unmatched list $\mathcal{L}_{U}$, if $\mathcal{L}_{p,n}\neq\emptyset$, $n$ requests its first preferred SF and removes it from $\mathcal{L}_{p,n}$, otherwise the end-device is removed from $\mathcal{L}_{U}$ since all SFs it can use have already reached their quota. Then, each $\mathrm{SF}_m$ either accepts all current requests if its quota allows it, or it accepts the requests of its most preferred end-devices that fulfill its quota, if not. This process is repeated until $\mathcal{L}_{U}$ becomes empty.\\

\textbf{Matching Refinement:} for each matching pair ($n,m$), the algorithm calculates $U_m$ using~(\ref{eq:utilitySF_1}) if it is only assigned to end-device $n$ and~(\ref{eq:utilitySF_2}) in the other case. The utility of end-device $n$ is calculated by~(\ref{eq:utilityDev_1}) if it is the only one at $\mathrm{SF}_m$, and with~(\ref{eq:utilityDev_2}) otherwise. Firstly, if there is an $\mathrm{SF}_l$ that is not assigned to any end-device that allows to increase $U_n$, the end-device leaves $\mathrm{SF}_m$ to be matched with $\mathrm{SF}_l$. Then, the algorithm calculates the utilities of every pair ($k,l$), and makes a swap between ($n,m$) and ($k,l$) and determines their new utilities. Secondly, if ($k,m$) or ($n,l$) is a blocking pair, the algorithm makes a swap. This swapping step is repeated until reaching a two-sided exchange stable matching.

\begin{algorithm}
%\resizebox{5\textwidth}{!}
\small{
\caption{Initial Matching}
\label{algo:Initial_matching}
\textbf{Initialization: } Set of unmatched end-devices: $\mathcal{L}_U\gets\mathcal{N}$, $\mathcal{A}_m\gets\emptyset$ 
\begin{algorithmic}[1]
\While{$\mathcal{L}_U \neq \emptyset$} 
\For{$i \in \mathcal{L}_U$}
\If{$ \mathcal{L}_{p,i}=\emptyset$}
\State $\mathcal{L}_U\gets \mathcal{L}_U \backslash \{i\} $;
\Else 
\State $a\gets\mathrm{firstPrefered(\mathcal{L}_{p,i})}$; \Comment{Favorite SF}
\State $\mathcal{L}_{p,i}\gets \mathcal{L}_{p,i} \backslash \{a\}$;
\State $\mathrm{req}_a\gets \mathrm{req}_a \cup \{i\}$;
\EndIf
\EndFor
\For{$j \in \mathcal{M}$}
\If{$\mathit{size}(\mathrm{req}_j)> 0  \And \mathit{size}(\mathcal{A}_j)< N_{Max}(j)$} \If{$(\mathit{size}(\mathrm{req}_j) +\mathit{size}(\mathcal{A}_j)) \leq  N_{Max}(j)$}
\State Accept all the requests and add the end-devices to $\mathcal{A}_j$;
\Else
\State Accept the requests of the $\left(\mathrm{N_{max}}-\mathit{size}(\mathcal{A}_j)\right)$ most preferred end-devices;
\State Add them to $\mathcal{A}_j$;
\EndIf
\EndIf
\EndFor
\EndWhile
\end{algorithmic}
}
\end{algorithm}
\begin{algorithm}
%\resizebox{.8\textwidth}{!}
\small{
\caption{Matching Refinement}
\label{algo:Matching_refinement}
\begin{algorithmic}[1]
\State change$\gets$true; 
\While{$\mathrm{change} = \mathrm{true}$} \State change$\gets$false; 
\For{$ j \in \mathcal{M}$}
\State Calculate $\mathrm{U}_j$; \Comment{eq.~(\ref{eq:utilitySF_1}) or eq.~(\ref{eq:utilitySF_2})}
\For{$ i \in \mathcal{A}_j$}
\State Calculate $\mathrm{U}_i$; \Comment{eq.~(\ref{eq:utilityDev_1}) or eq.~(\ref{eq:utilityDev_2})}

\For{$l \in \mathcal{M}_{-j}$}
\If{$\mathit{size}(\mathcal{A}_l)=0$} \State Swap\big($(i,j)$,$(\emptyset,l)$\big);
\State Calculate the new utility $\mathrm{U}^\prime_i$ of $i$; \Comment{eq.~(\ref{eq:utilityDev_1}) or eq.~(\ref{eq:utilityDev_2})} \If{$\mathrm{U}^\prime_i \geq \mathrm{U}_i$}
\State Validate the Swap;
\State $\mathrm{change}\gets\mathrm{true}$;
\EndIf
\Else
\State Calculate $\mathrm{U}_l$;  \Comment{eq.~(\ref{eq:utilitySF_1}) or eq.~(\ref{eq:utilitySF_2})}
\For{$ k \in \mathcal{A}_l$}
\State Calculate $\mathrm{U}_k$;  \Comment{eq.~(\ref{eq:utilityDev_1}) or eq.~(\ref{eq:utilityDev_2})} \State Swap\big($(i,j)$,$(k,l)$\big);
\If{$(i,l)$ or $(k,j)$ is a blocking pair} 
\State Validate the Swap;
\State $\mathrm{change} \gets \mathrm{true}$; 
\EndIf
\EndFor
\EndIf
\EndFor
\EndFor
\EndFor
\EndWhile
\end{algorithmic}
}
\end{algorithm}

\subsection{Proposed SF-Allocation Algorithm Analysis}\label{algo_analysis}
We now prove the stability and convergence of the proposed SF-Allocation algorithm, and analyze its computational complexity.
\newtheorem{theo}{Proposition}
\begin{theo}{\textbf{Stability}:}
When the proposed algorithm terminates, it finds a two-sided exchange stable matching.
\end{theo}
\begin{proof}
Let us assume that the proposed SF-allocation algorithm terminates and the final matching is not two-sided exchange stable. Then, the matching contains at least one more blocking pair $(k,m)$ or $(n,l)$ where the utility of at least one player among $\{n,m,k,l\}$, can be improved without lowering the others' utility. Accordingly, the proposed algorithm would continue, thereby the matching would not be final, which contradicts the initial assumption.
\end{proof}
%\vspace{-1cm}
\begin{theo}{\textbf{Convergence}:}
After a finite number of swap operations, the algorithm eventually converges to a two-sided exchange stable matching.
\end{theo}
\begin{proof}
A swap operation occurs if it improves the utility of at least one player without decreasing the others', hence the utilities can only rise. Additionally, the maximal throughput that can be achieved on an $\mathrm{SF}_m$ is upper-bounded by the data bit-rate $R_m$, meaning that each $\mathrm{SF}_m$ and the end-devices assigned to it have utilities upper bounded by $R_m$.\\
The number of potential swap operations is finite: end-device assigned to $\mathrm{SF}_{l}$ can make at most $\mathrm{N_{max}}(l) \times \sum\limits_{j\in \mathcal{M}_{-l}}\mathrm{N_{max}}(j)\normalsize$ swap operations. The total number of swap operations is thus upper-bounded by $ \sum\limits_{l\in\mathcal{M}}\mathrm{N_{max}}(l)\times \sum\limits_{j \in \mathcal{M}_{-l}}\mathrm{N_{max}}(j)\normalsize$.
\end{proof}
\begin{theo}{\textbf{Complexity}:}
The running time of our proposed algorithm is upper-bounded by $\mathcal{O}\left(NM + Q^2M^2\right)$, where $Q = \max\limits_{m\in \mathcal{M}}\{\mathrm{N_{max}}(m)\}$.
\end{theo}
\begin{proof}
\textit{Initial matching complexity}:
in the worst case, all the end-devices have the same preference list, and they are located in the area covered by all the SFs. At round$_1$ the gateway receives $N$ requests, at round$_2$ it receives $ N-\mathrm{N_{max}}(m_{1})\normalsize$ requests, at round$_i$ it receives $\footnotesize N-\sum\limits_{k=1}^{M-1}\mathrm{N_{max}}(m_{i}) \normalsize$ requests. Therefore, the total number of requests equals
$ NM-\sum\limits_{i=1}^{M-1}(M-i)\times \mathrm{N_{max}}(m_{i})\normalsize$.
The complexity of the initial matching is upper bounded by:
$\mathcal{O}\big(NM)$.\\
\textit{Matching refinement complexity}:
in each iteration, for each $\mathrm{SF}_m$, the algorithm considers at most $\mathrm{N_{max}}(m)$ end-devices and examines $ \sum\limits_{l\in \mathcal{M}_{-m}}\mathrm{N_{max}}(l)\normalsize$ swap operations for each of these end-devices. Therefore, the number of swap operations that are examined in one iteration is upper bounded by $ \sum\limits_{m\in \mathcal{M}}\mathrm{N_{max}}(m)\times\sum\limits_{l\in \mathcal{M}_{-m}}\mathrm{N_{max}}(l)\normalsize$.
Let $Q = \max\limits_{m\in \mathcal{M}}\{\mathrm{N_{max}}(m)\}$, thus the computational complexity of the matching refinement is upper bounded by $\mathcal{O}\big(Q^2M(M-1))$.\\
In summary, the computational complexity of our algorithm is upper bounded by $\mathcal{O}\left(NM + Q^2M^2\right)$.
\end{proof}
Note that this complexity is not excessive as our algorithm is run at the gateway which is not computationally-limited.

\section{Proposed Power Allocation Optimization}
\label{sec:tp-allocation}
Once the end-devices are assigned to SFs, we next optimize the power allocation variables in order to maximize the minimal throughput achieved on each SF. Given the fixed assignment variables $s_{nm}, \forall n,m$ from the previous step, the power allocation problem can be written as follows,

\begin{subequations}\label{obj_fun}
\setlength{\jot}{0.1pt}
\begin{alignat}{2} 
\nonumber&\max_{{p_{nm}}}\min_{\substack{(n,m)\in \\  \mathcal{N}\times \mathcal{M}}}f(p_{nm})=R_m \Bigg[ I\left({\sum\limits_{k \in \mathcal{N}}s_{km}=1}\right) \mathrm{e}^{-\frac{\mathrm{\tilde{\theta}}_{m}\sigma_c^2r_n^\alpha}{\mathrm{A}\left(f_c\right)p_{nm}}}\prod\limits_{i \in \mathcal{N}_{-n} }\prod\limits_{j \in \mathcal{M}_{-m}}\frac{1}{\mathrm{\tilde{\theta}}_{m}\frac{p_{ij}}{p_{nm} } {\left(\frac{r_n}{r_i}\right)}^\alpha+1}\\&+ I\left({\sum\limits_{k \in \mathcal{N}}s_{km}\geq2}\right)\mathrm{e}^{-\frac{\mathrm{\theta_{co}}\sigma_c^2r_n^\alpha}{ A\left(f_c\right)p_{nm}}}\left(\prod\limits_{i \in \mathcal{N}_{-n}}\prod\limits_{j \in \mathcal{M}_{-m}} \frac{1}{\mathrm{\theta_{co}}\frac{p_{ij}}{ p_{nm}} {\left(\frac{r_n}{r_i}\right)}^\alpha+1}\right)\prod\limits_{i \in \mathcal{N}_{-n}} \frac{1}{\mathrm{\theta_{co}}\frac{p_{im}}{ p_{nm}} {\left(\frac{r_n}{r_i}\right)}^\alpha+1} \Bigg]\tag{\ref{obj_fun}}\\
&\text{s.t.}\quad \text{C1: } 0 \leq  p_{nm} \leq P_\mathrm{max}, p_{nm}\in\mathbb{R}^+ \label{opt:cond_0}
\end{alignat}
\normalsize 
\end{subequations}

It can be observed that the objective function $f(p_{nm})$ of problem~(\ref{obj_fun_sf}), unlike in previous works such as~\cite{Zhi17dec}, is non-linear non-convex, for which a global optimum is difficult to obtain. This greatly increases the difficulty of this optimization problem. Instead, we seek for a near-optimal solution by transforming the initial problem as follows. Let $\mathcal{P}_\eta$ be the set of transmit power vectors $\mathbf{p}$ such that the minimum throughput over end-devices and SFs is above a certain parameter $\eta \in \mathbb{R}$, namely 
\begin{equation}
\setlength{\jot}{0.1pt}
    \mathcal{P}_\eta=\left\{\mathbf{p}|\min_{m}f\left(p_{nm}\right) \geq \eta , \forall n \in \mathcal{N}\right\}.
\end{equation}

Since the minimal throughput value is above $\eta$, all throughput values should be above $\eta$ as well. Hence, defining
\begin{equation}
\setlength{\jot}{0.1pt}
    \mathcal{P}^*_\eta=\left\{\mathbf{p}|f\left(p_{nm}\right) \geq \eta , \forall n \in \mathcal{N}, \forall m \in \mathcal{M}\right\},
\end{equation}
we can write $\mathcal{P}^*_\eta = \mathcal{P}_\eta$. Introducing a new variable $\eta \in\mathbb{R}^+$, problem~(\ref{obj_fun}) is equivalent to the following optimization problem,
\begin{subequations}\label{obj_fun_p}
\setlength{\jot}{0.1pt}
\begin{alignat}{2} 
    \max_{p_{nm}, \eta} &\quad \eta
    \tag{\ref{obj_fun_p}}\\
    \text{s.t.}
&\quad \text{C1: } 0 \leq  p_{nm} \leq P_\mathrm{max}, p_{nm}\in\mathbb{R}^+\\
&\quad\text{C2: } \mathbf{p}\in \mathcal{P}^*_\eta
\label{opt:cond_0_p}
\end{alignat}
\normalsize 
\end{subequations}

Therefore, we take the following approach: for a given $\eta$, we solve the feasibility problem
%\vspace{-.5cm}
\begin{subequations}\label{obj_fun_p_1}
\setlength{\jot}{0.1pt}
\begin{alignat}{2} 
    \text{Find} &\quad \mathbf{p}
    \tag{\ref{obj_fun_p_1}}\\
    &\text{s.t.}
    & \quad \mathbf{p} \in [0,P_\mathrm{max}]^{NM\times 1} \cap \mathcal{P}^*_\eta, \label{opt:cond_0_p_o}
\end{alignat}
\normalsize 
\end{subequations}
then $\eta$ is increased until no feasible $\mathbf{p}$ can be found. In practice, parameter $\eta$ can be updated using the bisection method~\cite{Zhi17dec} as detailed in Algorithm~\ref{algo:power_1}, as follows. Initially, $\eta$ is lower-bounded by $\eta_{min}=0$, upper-bounded by $\eta_{max}$ which is equal to the minimal bit-rate over allocated SFs and end-devices. First, setting $\eta$ as the midpoint of the interval $[\eta_{min}, \eta_{max}]$, problem~\eqref{obj_fun_p_1} is solved and if a feasible solution is found, it is denoted as $\mathbf{p}_{opt}$ and we update the lower bound $\eta_{min}$ as $\eta$. Otherwise, if no feasible power vector is found, $\eta_{max}$ is set as $\eta$. This procedure is iterated until the interval length $[\eta_{min}, \eta_{max}]$ is smaller than the desired accuracy $\epsilon$.

\begin{algorithm}
\small
\caption{Power allocation optimization}
\label{algo:power_1}
\textbf{Initialization: } $\eta_{min}\gets 0$, $\eta_{max}\gets \min\limits_{m \in \mathcal{M}} R_m$, $\epsilon>0$.
\begin{algorithmic}[1]
\While{$\eta_{max}-\eta_{min} \geq \epsilon $}
\State $\eta \gets \frac{(\eta_{max}+\eta_{min})}{2}$;
\State Solve~\eqref{obj_fun_p_1}: find a transmit power vector $\mathbf{p}$ satisfying the constraint in~(\ref{obj_fun_p_1});
\If{$\mathbf{p}$ $ \mathrm{exists}$}
\State $\mathbf{p}_{opt}\gets \mathbf{p}$;
\State Calculate the utilities of each $\mathrm{SF}_m$, $\mathrm{U}_m$ using $\mathbf{p}_{opt}$
\State $\eta_{min} \gets \eta$;
\Else
\State $\eta_{max} \gets \eta$;
\EndIf
\State \hspace{-.7cm} $\mathcal{P}^*_\eta \gets \mathbf{p}_{opt}$;
\EndWhile
\end{algorithmic}
\end{algorithm}

However, $\mathcal{P}^*_\eta$ contains non-linear inequalities, making it difficult to solve the feasibility problem~\eqref{obj_fun_p_1}. Hence, we devise two methods for making this problem tractable: linear approximation (A) and quadratic approximation (B) of these non-linear inequalities.  

\subsection{Feasibility problem with linear approximation }
In this subsection, in order to make problem~\eqref{obj_fun_p_1} tractable, we first approximate the non-linear inequalities in the set $\mathcal{P}^*_\eta$ by linear ones. We distinguish two cases, one where only a single end-device is assigned to $\mathrm{SF}_m$ and the second, where more than one end-devices are assigned to $\mathrm{SF}_m$.
\subsubsection{Case 1}
a single end-device $n$ is assigned to $\mathrm{SF}_m$, hence $n$ is only subject to inter-SF interferences. Therefore, given~\eqref{cap_inter-SF_5}, $\mathcal{P}^*_\eta$ is given by,

\begin{equation}
\setlength{\jot}{0.1pt}
\begin{aligned}
    \mathcal{P}^*_\eta=\Bigg\{\mathbf{p}\Big|R_m  \mathrm{e}^{-\frac{\mathrm{\tilde{\theta}}_{m}\sigma_c^2r_n^\alpha}{\mathrm{A}\left(f_c\right)p_{nm}}}\prod\limits_{i \in \mathcal{N}_{-n} }\prod\limits_{j \in \mathcal{M}_{-m}}\frac{1}{\mathrm{\tilde{\theta}}_{m}\frac{p_{ij}}{p_{nm} } {\left(\frac{r_n}{r_i}\right)}^\alpha+1}\geq \eta, \, \forall m \in \mathcal{M}\Bigg\}.
\end{aligned}
\label{app_lin_c1_0}
\end{equation}

Rearranging and taking the logarithm of both sides, the inequalities in~\eqref{app_lin_c1_0} are equivalent to

\begin{equation}
\setlength{\jot}{0.1pt}
    \frac{\mathrm{\tilde{\theta}}_{m}\sigma_c^2r_n^\alpha}{\mathrm{A}\left(f_c\right)p_{nm}} +\sum\limits_{i \in  \mathcal{N}_{-n}} \sum\limits_{j \in \mathcal{M}_{-m} }\mathrm{ln}\left(\mathrm{\tilde{\theta}}_{m}\frac{p_{ij}}{p_{nm} } {\left(\frac{r_n}{r_i}\right)}^\alpha+1\right)\leq -\mathrm{ln}\left(\frac{\eta}{R_m }\right), \forall m \in \mathcal{M} .
\label{app_lin_c1_2}
\end{equation}

The term $\mathrm{\tilde{\theta}}_m\frac{p_{ij}}{p_{nm} } {\left(\frac{r_n}{r_i}\right)}^\alpha$ is dominated by the inter-SF interference capture threshold $\mathrm{\tilde{\theta}}_m$, which takes very small values as can be observed from Table~\ref{tab:tab1}. Thus, the term $\mathrm{\tilde{\theta}}_{m}\frac{p_{ij}}{p_{nm} } {\left(\frac{r_n}{r_i}\right)}^\alpha$ will be generally close to zero, as confirmed by the numerical evaluations in Section~\ref{sec:num_res}. Therefore, we can approximate the logarithmic term using the Taylor-Maclaurin series,
\begin{equation}
\setlength{\jot}{0.1pt}
\mathrm{ln}\left(\mathrm{\tilde{\theta}}_{m}\frac{p_{ij}}{p_{nm} } {\left(\frac{r_n}{r_i}\right)}^\alpha+1\right) = \mathrm{\tilde{\theta}}_{m}\frac{p_{ij}}{p_{nm} } {\left(\frac{r_n}{r_i}\right)}^\alpha + \mathrm{o}\left(\mathrm{\tilde{\theta}}_{m}\frac{p_{ij}}{p_{nm} } {\left(\frac{r_n}{r_i}\right)}^\alpha\right),
\label{app_lin_c1_3}
\end{equation}
where $\mathrm{o}\left(\mathrm{\tilde{\theta}}_{m}\frac{p_{ij}}{p_{nm} }{\left(\frac{r_n}{r_i}\right)}^\alpha\right)$ denotes the remainder of the Taylor series.

By substituting $\mathrm{ln}\left(\mathrm{\tilde{\theta}}_{m}\frac{p_{ij}}{p_{nm} } {\left(\frac{r_n}{r_i}\right)}^\alpha+1\right)$ by its approximation~\eqref{app_lin_c1_3} in~\eqref{app_lin_c1_2} and rearranging, we get the following linear inequalities,
\begin{equation}
\setlength{\jot}{0.1pt}
     \mathrm{ln}\left(\frac{\eta}{R_m }\right)p_{nm} +  \sum\limits_{i \in  \mathcal{N}_{-n}} \sum\limits_{j \in\mathcal{M}_{-m} }
    \mathrm{\tilde{\theta}}_{m} {\left(\frac{r_n}{r_i}\right)}^\alpha p_{ij}\leq -\frac{\mathrm{\tilde{\theta}}_{m}\sigma_c^2r_n^\alpha}{\mathrm{A}\left(f_c\right)}, \forall m \in \mathcal{M}.
\label{app_lin_c1_5}
\end{equation}

\subsubsection{Case 2}
\label{case2_app}
if $\mathrm{SF}_m$ is shared by more than one end-device, from~\eqref{cap_co-sf_3}, the set $\mathcal{P}^*_\eta$ is given by

\begin{equation}
\setlength{\jot}{0.1pt}
\begin{aligned}
    \mathcal{P}^*_\eta=\Bigg\{\mathbf{p}\Big| R_m \mathrm{e}^{-\frac{\mathrm{\theta_{co}}\sigma_c^2r_n^\alpha}{ A\left(f_c\right)p_{nm}}}\left(\prod\limits_{i \in \mathcal{N}_{-n}}\prod\limits_{j \in \mathcal{M}_{-m}} \frac{1}{\mathrm{\theta_{co}}\frac{p_{ij}}{ p_{nm}} {\left(\frac{r_n}{r_i}\right)}^\alpha+1}\right)\prod\limits_{i \in \mathcal{N}_{-n}} \frac{1}{\mathrm{\theta_{co}}\frac{p_{im}}{ p_{nm}} {\left(\frac{r_n}{r_i}\right)}^\alpha+1} \geq \eta , \, \forall m \in \mathcal{M}\Bigg\}.
\end{aligned}
\label{app_lin_c2_0}
\end{equation}

Similarly to Case 1, we perform the following linearization in order to make problem~\eqref{obj_fun_p_1} tractable. By rearranging the inequalities, we obtain for all $m \in \mathcal{M}$,\\ 
%\begin{equation}
%    \mathrm{e}^{-\frac{\mathrm{\theta_{co}}\sigma_c^2r_n^\alpha}{ A\left(f_c\right)p_{nm}}}\prod\limits_{i \in \mathcal{N}_{-n}}\prod\limits_{j \in \mathcal{M}_{-m}} \frac{1}{\mathrm{\theta_{co}}\frac{p_{ij}}{ p_{nm}} {\left(\frac{r_n}{r_i}\right)}^\alpha+1}\prod\limits_{i \in \mathcal{N}_{-n}} \frac{1}{\mathrm{\theta_{co}}\frac{p_{im}}{ p_{nm}} {\left(\frac{r_n}{r_i}\right)}^\alpha+1} \geq \frac{\lambda}{R_m}
%    \label{app_lin_c2_1}
%\end{equation}
\begin{equation}
\setlength{\jot}{0.1pt}
    {-\frac{\mathrm{\theta_{co}}\sigma_c^2r_n^\alpha}{ A\left(f_c\right)p_{nm}}}-\sum\limits_{i \in \mathcal{N}_{-n}}\sum\limits_{j \in \mathcal{M}_{-m}} \mathrm{ln}\left(\mathrm{\theta_{co}}\frac{p_{ij}}{ p_{nm}} {\left(\frac{r_n}{r_i}\right)}^\alpha+1\right)-\sum\limits_{i \in \mathcal{N}_{-n}} \mathrm{ln}\left(\mathrm{\theta_{co}}\frac{p_{im}}{ p_{nm}} {\left(\frac{r_n}{r_i}\right)}^\alpha+1\right) \geq \mathrm{ln}\left(\frac{\eta}{R_m}\right).
    \label{app_lin_c2_2}
\end{equation}

However, in this case, the co-SF interference capture threshold $\mathrm{\theta_{co}}$ no longer induces small values of $\mathrm{\theta_{co}}\frac{p_{im}}{ p_{nm}} {\left(\frac{r_n}{r_i}\right)}^\alpha$, since in practice, $\mathrm{\theta_{co}}=6$ dB~\cite{geo17apr}. Therefore, we now make use of a different approximation based on Taylor's theorem.

Let $g(x)=\mathrm{ln}\left(\frac{\mathrm{\theta_{co}}}{ p_{nm}} {\left(\frac{r_n}{r_i}\right)}^\alpha x+1\right)$. Clearly, $g$ is a twice continuously differentiable function.

From Taylor's theorem, we have
\begin{equation}
g(x)= g(a)+g^\prime(a)(x-a)+\mathrm{o}\left(x-a\right), \forall a \in \mathbb{R}^+.
\label{app_lin_c2_3}
\end{equation}

Taking $a=\left(\frac{r_i}{r_n}\right)^\alpha\frac{p_{nm}}{\mathrm{\theta_{co}}}$ and given
\begin{equation}
\setlength{\jot}{0.1pt}
g'(x)=\frac{\frac{\mathrm{\theta_{co}}}{ p_{nm}} {\left(\frac{r_n}{r_i}\right)}^\alpha}{\frac{\mathrm{\theta_{co}}}{ p_{nm}} {\left(\frac{r_n}{r_i}\right)}^\alpha x+1}=\frac{1}{x+a},
\label{app_lin_c2_4}
\end{equation}
~\eqref{app_lin_c2_3} may be written
\begin{equation}
\setlength{\jot}{0.1pt}
g(x)= \left(\mathrm{ln}\left(2 \right)-\frac{1}{2}\right)+\frac{\mathrm{\theta_{co}}}{2p_{nm}} {\left(\frac{r_n}{r_i}\right)}^\alpha x +\mathrm{o}\left(x-\left(\frac{r_i}{r_n}\right)^\alpha\frac{p_{nm}}{\theta_{co}}\right).
\label{app_lin_c2_5}
\end{equation}
%The first derivative of $f$ is as follows
%\begin{equation}f^\prime\left(p_{im}\right)=\frac{\frac{\mathrm{\theta_{co}}}{ p_{nm}} {\left
%(\frac{r_n}{r_i}\right)}^\alpha}{\mathrm{\theta_{co}}\frac{p_{im}}{ p_{nm}} {\left
%(\frac{r_n}{r_i}\right)}^\alpha+1}\label{app_lin_c2_5}\end{equation}
%\begin{equation}
%\small
%\begin{aligned}
%f(p_{im})&=\mathrm{ln}\left(\frac{\mathrm{\theta_{co}}}{ p_{nm}} {\left(\frac{r_n}{r_i}\right)}^\alpha\times\left(\frac{r_i}{r_n}\right)^\alpha\frac{p_{nm}}{\theta_{co}} +1\right)+\frac{\frac{\mathrm{\theta_{co}}}{ p_{nm}} {\left(\frac{r_n}{r_i}\right)}^\alpha}{\frac{\mathrm{\theta_{co}}}{ p_{nm}} {\left(\frac{r_n}{r_i}\right)}^\alpha \times\left(\left(\frac{r_i}{r_n}\right)^\alpha\frac{p_{nm}}{\theta_{co}} +1\right)}\left(p_{im}-\left(\frac{r_i}{r_n}\right)^\alpha\frac{p_{nm}}{\theta_{co}}\right)\\
%&=\mathrm{ln}(2)+\frac{\frac{\mathrm{\theta_{co}}}{ p_{nm}} {\left(\frac{r_n}{r_i}\right)}^\alpha}{2}p_{im}-\frac{\frac{\mathrm{\theta_{co}}}{ p_{nm}} {\left(\frac{r_n}{r_i}\right)}^\alpha}{2}\left(\frac{r_i}{r_n}\right)^\alpha\frac{p_{nm}}{\theta_{co}}\\
%&=\mathrm{ln}(2)+\frac{\mathrm{\theta_{co}}}{ 2p_{nm}} {\left(\frac{r_n}{r_i}\right)}^\alpha p_{im}-\frac{1}{2}\\
%\end{aligned}
%\label{app_lin_c2_6}
%\end{equation}

Dropping the remainder $\mathrm{o}$ and substituting the logarithmic terms of~\eqref{app_lin_c2_2} by their linear expressions in~\eqref{app_lin_c2_5} and taking $x=p_{ij}$, we obtain the linearized expressions 
\begin{equation}
\begin{aligned}
    &\frac{\mathrm{\theta_{co}}\sigma_c^2r_n^\alpha}{ \mathrm{A}\left(f_c\right)p_{nm}}+\left(\sum\limits_{i \in\mathcal{N}_{-n}}\sum\limits_{j \in\mathcal{M}_{-m}}\mathrm{ln}(2)-\frac{1}{2}+\frac{\mathrm{\theta_{co}}}{ 2} {\left(\frac{r_n}{r_i}\right)}^\alpha \frac{p_{ij}}{p_{nm}}\right)\\&+\sum\limits_{i \in\mathcal{N}_{-n}} \mathrm{ln}(2)-\frac{1}{2}+\frac{\mathrm{\theta_{co}}}{ 2} {\left(\frac{r_n}{r_i}\right)}^\alpha \frac{p_{im}}{p_{nm}} \leq -\mathrm{ln}\left(\frac{\eta}{R_m}\right).
    \label{app_lin_c2_8}
\end{aligned}
\end{equation}

Finally, from~\eqref{app_lin_c1_5} and~\eqref{app_lin_c2_8}, problem~\eqref{obj_fun_p_1} can be expressed as,

\begin{subequations}\label{obj_fun_app_lin}
\vspace{-.5cm}
\setlength{\jot}{0.1pt}
\begin{alignat}{2}
 \text{ Find }& \mathbf{p} \tag{\ref{obj_fun_app_lin}}\\
\text{s.t. }& \text{C1: }0 \leq  p_{nm} \leq P_\mathrm{max} \label{obj_fun_app_lin:cond_0}\\
 & \text{C6: }\mathrm{ln}\left(\frac{\eta}{R_m }\right)p_{nm} + \frac{\mathrm{\tilde{\theta}}_{m}\sigma_c^2r_n^\alpha}{\mathrm{A}\left(f_c\right)}+\sum\limits_{i \in  \mathcal{N}_{-n}}\sum\limits_{j \in \mathcal{M}_{-m} }
    \mathrm{\tilde{\theta}}_{m} {\left(\frac{r_n}{r_i}\right)}^\alpha p_{ij}\leq 0, \text{ if} \sum\limits_{k \in \mathcal{N}}s_{km}=1 \label{obj_fun_app_lin:cond_1} \\
 & \text{C7: }\mathrm{ln}\left(\frac{\eta}{R_m}\right)p_{nm}+{\frac{\mathrm{\theta_{co}}\sigma_c^2r_n^\alpha}{ A\left(f_c\right)}}+\left(\sum\limits_{i \in \mathcal{N}_{-n}}\sum\limits_{j \in \mathcal{M}_{-m}} \left(\mathrm{ln}\left(2\right)-\frac{1}{2}\right)p_{nm}+\frac{\mathrm{\theta_{co}}}{2} {\left(\frac{r_n}{r_i}\right)}^\alpha\ p_{ij}\right)+\nonumber\\&\sum\limits_{i \in \mathcal{N}_{-n}} \left(\mathrm{ln}\left(2\right)-\frac{1}{2}\right)p_{nm}+\frac{\mathrm{\theta_{co}}}{2} {\left(\frac{r_n}{r_i}\right)}^\alpha p_{im}\leq 0, \text{ if} \sum\limits_{k \in \mathcal{N}}s_{km}\geq 2 \label{obj_fun_app_lin:cond_2}
\end{alignat}
\normalsize 
\end{subequations}

Although problem~\eqref{obj_fun_app_lin} is a feasibility problem, by taking an arbitrary objective function, it can be written as a linear programming problem defined by linear inequalities. Hence, it can be solved with usual linear programming solvers such as \textit{linprog} or \textit{fmincon} in Matlab.

\subsection{Feasibility problem with quadratic approximation}
In this subsection, we propose a second method for making problem~\eqref{obj_fun_p_1} tractable, by means of quadratic approximation of the non-linear inequalities of $\mathcal{P}^*_\eta$. As in the previous subsection, we distinguish two cases:
\subsubsection{Case 1} only one end-device assigned to $\mathrm{SF}_m$. $\mathcal{P}^*_\eta$ is equal to~\eqref{app_lin_c1_0}. The quadratic approximation of the logarithmic terms in~\eqref{app_lin_c1_2} using the Taylor-Maclaurin series, is given by
\begin{equation}
\mathrm{ln}\left(\mathrm{\tilde{\theta}}_{m}\frac{p_{ij}}{p_{nm} } {\left(\frac{r_n}{r_i}\right)}^\alpha+1\right)= \mathrm{\tilde{\theta}}_{m}\frac{p_{ij}}{p_{nm} } {\left(\frac{r_n}{r_i}\right)}^\alpha -\frac{\mathrm{\tilde{\theta}}_{m}^2}{2}\left(\frac{p_{ij}}{p_{nm} }\right)^2 {\left(\frac{r_n}{r_i}\right)}^{2\alpha}+\mathrm{o}\left(\mathrm{\tilde{\theta}}_{m}^2\left(\frac{p_{ij}}{p_{nm} }\right)^2 {\left(\frac{r_n}{r_i}\right)}^{2\alpha}\right).
\label{app_qua_c1_0}
\end{equation}

Substituting the logarithmic term in~\eqref{app_lin_c1_2} and rearranging, we obtain the following inequality,
\begin{equation}
\setlength{\jot}{0.1pt}
     \mathrm{ln}\left(\frac{\eta}{R_m }\right)p_{nm}^2+\frac{\mathrm{\tilde{\theta}}_{m}\sigma_c^2r_n^\alpha}{\mathrm{A}\left(f_c\right)}p_{nm}+\sum\limits_{i \in  \mathcal{N}_{-n}} \sum\limits_{j \in \mathcal{M}_{-m} }
    \mathrm{\tilde{\theta}}_{m} {\left(\frac{r_n}{r_i}\right)}^\alpha p_{ij}p_{nm}-\frac{\mathrm{\tilde{\theta}}_{m}^2}{2} {\left(\frac{r_n}{r_i}\right)}^{2\alpha}p_{ij}^2\leq 0.
\label{app_qua_c1_2}
\end{equation}

\subsubsection{Case 2} $\mathrm{SF}_m$ is shared by more than one end-device. $\mathcal{P}^*_\eta$ is given by~\eqref{app_lin_c2_0}.\\
As in Section~\ref{case2_app} let $g(x)=\mathrm{ln}\left(\mathrm{\theta_{co}}\frac{x}{ p_{nm}} {\left(\frac{r_n}{r_i}\right)}^\alpha+1\right)$. $g$ is a twice continuously differentiable function. From Taylor's theorem, we have,
\begin{equation}
\setlength{\jot}{0.1pt}
g(x)= g(a)+g^\prime(a)(x-a)+\frac{g^{\prime\prime }(a)}{2!}(x-a)^2+\mathrm{o}\left((x-a)^2\right), \forall a \in \mathbb{R}^+.
\label{app_qud_c2_3}
\end{equation}

The first derivative of $g$ is given in~\eqref{app_lin_c2_4}, and its second derivative by

\begin{equation}
g''(x)=-\frac{\frac{\mathrm{\theta_{co}}^2}{ p_{nm}^2} {\left(\frac{r_n}{r_i}\right)}^{2\alpha}}{\left(\mathrm{\theta_{co}}\frac{x}{ p_{nm}} {\left(\frac{r_n}{r_i}\right)}^\alpha+1\right)^2}=\frac{-1}{(x+a)^2}.
\label{app_qua_c2_1}
\end{equation}

From~\eqref{app_qud_c2_3}, the quadratic approximation of $g$ centered at $a=\left(\frac{r_i}{r_n}\right)^\alpha\frac{p_{nm}}{\mathrm{\theta_{co}}}$ and for $x=p_{ij}$, is given by,
%\begin{equation}
%f(a)=\mathrm{ln}\left(2\right)
%\label{app_qua_c2_2}
%\end{equation}
%\begin{equation}
%f'(p_{ij})=\frac{\mathrm{\theta_{co}}}{ 2p_{nm}} {\left(\frac{r_n}{r_i}\right)}^\alpha
%\label{app_qua_c2_3}
%\end{equation}
%\begin{equation}
%f''(p_{ij})=-\frac{\mathrm{\theta_{co}}^2}{ 4p_{nm}^2} {\left(\frac{r_n}{r_i}\right)}^{2\alpha}
%\label{app_qua_c2_4}
%\end{equation}
%Therefore,
%\begin{equation}
%f(p_{ij})\approx \mathrm{ln}\left(2\right)+\frac{\mathrm{\theta_{co}}}{ 2p_{nm}} {\left(\frac{r_n}{r_i}\right)}^\alpha\left(p_{ij}-\left(\frac{r_i}{r_n}\right)^\alpha\frac{p_{nm}}{\theta_{co}}\right)-\frac{\mathrm{\theta_{co}}^2}{ 4p_{nm}^2} {\left(\frac{r_n}{r_i}\right)}^{2\alpha}\left(p_{ij}-\left(\frac{r_i}{r_n}\right)^\alpha\frac{p_{nm}}{\theta_{co}}\right)^2
%\label{app_qua_c2_5}
%\end{equation}
%\begin{equation}
%f(p_{ij})\approx \mathrm{ln}\left(2\right)+\frac{\mathrm{\theta_{co}}}{ 2} {\left(\frac{r_n}{r_i}\right)}^\alpha\frac{p_{ij}}{p_{nm}}-\frac{1}{2}-\frac{\mathrm{\theta_{co}}^2}{ 8} {\left(\frac{r_n}{r_i}\right)}^{2\alpha}\frac{p_{ij}^2}{p_{nm}^2}-\frac{1}{8}+\frac{\mathrm{\theta_{co}}}{ 4} {\left(\frac{r_n}{r_i}\right)}^{\alpha}\frac{p_{ij}}{p_{nm}}
%\label{app_qua_c2_6}
%\end{equation}
\begin{equation}
\setlength{\jot}{0.1pt}
g(p_{ij})= \left(\mathrm{ln} \text{ } 2 -\frac{5}{8}\right)+\frac{3}{4}\mathrm{\theta_{co}} {\left(\frac{r_n}{r_i}\right)}^\alpha\frac{p_{ij}}{p_{nm}}-\frac{\mathrm{\theta_{co}}^2}{ 8} {\left(\frac{r_n}{r_i}\right)}^{2\alpha}\frac{p_{ij}^2}{p_{nm}^2}+\mathrm{o}\left(\left(p_{ij}-\left(\frac{r_i}{r_n}\right)^\alpha\frac{p_{nm}}{\mathrm{\theta_{co}}}\right)^2\right),
\label{app_qua_c2_7}
\end{equation}
and for $x=p_{im}$,
\begin{equation}
\setlength{\jot}{0.1pt}
g(p_{im})= \left(\mathrm{ln} \text{ } 2-\frac{5}{8}\right)+\frac{3}{4}\mathrm{\theta_{co}} {\left(\frac{r_n}{r_i}\right)}^\alpha\frac{p_{im}}{p_{nm}}-\frac{\mathrm{\theta_{co}}^2}{ 8} {\left(\frac{r_n}{r_i}\right)}^{2\alpha}\frac{p_{im}^2}{p_{nm}^2}+\mathrm{o}\left(\left(p_{im}-\left(\frac{r_i}{r_n}\right)^\alpha\frac{p_{nm}}{\mathrm{\theta_{co}}}\right)^2\right).
\label{app_qua_c2_15}
\end{equation}

Finally by dropping the remainder~\eqref{app_lin_c2_2} becomes,
\begin{equation}
\setlength{\jot}{0.1pt}
\begin{split}
     &{\frac{\mathrm{\theta_{co}}\sigma_c^2r_n^\alpha}{ A\left(f_c\right)p_{nm}}}+\left(\sum\limits_{i \in \mathcal{N}_{-n}}\sum\limits_{j \in \mathcal{M}_{-m}} \left(\mathrm{ln} \text{ } 2-\frac{5}{8}\right)+\frac{3}{4}\mathrm{\theta_{co}} {\left(\frac{r_n}{r_i}\right)}^\alpha\frac{p_{ij}}{p_{nm}}-\frac{\mathrm{\theta_{co}}^2}{ 8} {\left(\frac{r_n}{r_i}\right)}^{2\alpha}\frac{p_{ij}^2}{p_{nm}^2}\right)+\\&\sum\limits_{i \in \mathcal{N}_{-n}} \left(\mathrm{ln} \text{ } 2-\frac{5}{8}\right)+\frac{3}{4}\mathrm{\theta_{co}} {\left(\frac{r_n}{r_i}\right)}^\alpha\frac{p_{im}}{p_{nm}}-\frac{\mathrm{\theta_{co}}^2}{ 8} {\left(\frac{r_n}{r_i}\right)}^{2\alpha}\frac{p_{im}^2}{p_{nm}^2} \leq -\mathrm{ln}\left(\frac{\eta}{R_m}\right).
     \end{split}
     \label{app_qua_c2_16}
\end{equation}

By multiplying both sides by $p^2_{nm}$ we obtain,
\begin{equation}
\setlength{\jot}{0.1pt}
\begin{split}
     &\mathrm{ln}\left(\frac{\eta}{R_m}\right)p_{nm}^2+{\frac{\mathrm{\theta_{co}}\sigma_c^2r_n^\alpha}{ A\left(f_c\right)}}p_{nm}+\left(\sum\limits_{i \in \mathcal{N}_{-n}}\sum\limits_{j \in \mathcal{M}_{-m}} \Bigg(\mathrm{ln} \text{ } 2-\frac{5}{8}\right)p_{nm}^2+\frac{3}{4}\mathrm{\theta_{co}} {\left(\frac{r_n}{r_i}\right)}^\alpha p_{ij}p_{nm}\\&-\frac{\mathrm{\theta_{co}}^2}{ 8} {\left(\frac{r_n}{r_i}\right)}^{2\alpha}p_{ij}^2\Bigg)+\sum\limits_{i \in \mathcal{N}_{-n}} \left(\mathrm{ln} \text{ } 2-\frac{5}{8}\right)p_{nm}^2+\frac{3}{4}\mathrm{\theta_{co}} {\left(\frac{r_n}{r_i}\right)}^\alpha p_{im}p_{nm}-\frac{\mathrm{\theta_{co}}^2}{ 8} {\left(\frac{r_n}{r_i}\right)}^{2\alpha}p_{im}^2\leq 0
     \end{split}
     \label{app_qua_c2_17}
\end{equation}

From~\eqref{app_qua_c1_2} and~\eqref{app_qua_c2_17}, problem~\eqref{obj_fun_p_1} can be expressed as\\
\vspace{-.5cm}
\begin{subequations}\label{obj_fun_app_quad}
\setlength{\jot}{0.1pt}
\begin{alignat}{2} 
 \text{ Find } \mathbf{p} \tag{\ref{obj_fun_app_quad}}\\\nonumber\\
\text{s.t } &\text{C1: }0 \leq  p_{nm} \leq P_\mathrm{max} \label{obj_fun_app_quad:cond_0}\\
& \text{C6: }\mathrm{ln}\left(\frac{\eta}{R_m }\right)p_{nm}^2+\frac{\mathrm{\tilde{\theta}}_{m}\sigma_c^2r_n^\alpha}{\mathrm{A}\left(f_c\right)}p_{nm}+  \sum\limits_{i \in  \mathcal{N}_{-n}}\sum\limits_{j \in \mathcal{M}_{-m} }
    \mathrm{\tilde{\theta}}_{m} {\left(\frac{r_n}{r_i}\right)}^\alpha p_{ij}p_{nm}\nonumber\\&-\frac{1}{2}\mathrm{\tilde{\theta}}_{m}^2 {\left(\frac{r_n}{r_i}\right)}^{2\alpha}p_{ij}^2\leq 0, \text{ if} \sum\limits_{k \in \mathcal{N}}s_{km}=1 \label{obj_fun_app_quad:cond_1}\\
 & \text{C7: }\mathrm{ln}\left(\frac{\eta}{R_m}\right)p_{nm}^2+{\frac{\mathrm{\theta_{co}}\sigma_c^2r_n^\alpha}{ A\left(f_c\right)}}p_{nm}+\Bigg (\sum\limits_{i \in \mathcal{N}_{-n}}\sum\limits_{j \in \mathcal{M}_{-m}} \left(\mathrm{ln} \text{ } 2-\frac{5}{8}\right)p_{nm}^2 \nonumber\\&+\frac{3}{4}\mathrm{\theta_{co}} {\left(\frac{r_n}{r_i}\right)}^\alpha p_{ij}p_{nm}-\frac{\mathrm{\theta_{co}}^2}{ 8} {\left(\frac{r_n}{r_i}\right)}^{2\alpha}p_{ij}^2\Bigg )+\sum\limits_{i \in \mathcal{N}_{-n}} \left(\mathrm{ln} \text{ } 2-\frac{5}{8}\right)p_{nm}^2 \nonumber\\&+\frac{3}{4}\mathrm{\theta_{co}} {\left(\frac{r_n}{r_i}\right)}^\alpha p_{im}p_{nm}-\frac{\mathrm{\theta_{co}}^2}{ 8} {\left(\frac{r_n}{r_i}\right)}^{2\alpha}p_{im}^2\leq 0, \text{ if} \sum\limits_{k \in \mathcal{N}}s_{km}\geq 2 \label{obj_fun_app_quad:cond_2}
\end{alignat}
\normalsize 
\end{subequations}

Problem~\eqref{obj_fun_app_quad} is a feasibility problem with quadratic inequality constraints. Hence, solutions can be computed by means of  solvers such \textit{fmincon} in Matlab.

\section{Numerical Results}
\label{sec:num_res}
\subsection{Simulation Settings}
We basically use the simulation parameters of references~\cite{geo17apr,War18}. Namely, we consider a cell of radius $R=1$ km, with a varying number of devices $N$ from 2 to 40. Note that all devices transmit with a duty cycle of 100\%. Hence, with a duty cycle of 1\% as preconized in LoRaWAN~\cite{loraA}, the actual number of end-devices would theoretically be 100-fold\footnote{The evaluations are made for 100\% duty cycle as this is the most challenging case. Hence, much better performance can be expected in the case of 1\% duty cycle.}, i.e., up to 4000. All end-devices transmit in the channel of carrier frequency $f_c=868$ MHz with a bandwidth $BW=125$ kHz. We consider a lossy urban environment, with a path loss exponent equal to 4. The maximal transmit power is fixed to $P_\mathrm{max}=14$ dBm. The number of iterations $N_I$ was fixed to 1, as it gives the best compromise between performance and computational complexity. 

\subsection{Baseline schemes}
We consider two baseline schemes for performance comparison: the random SF allocation~\cite{War18}, and the distance-SF allocation algorithms~\cite{mar16nov}, with a maximal number of simultaneously transmitting devices equal to $A=\sum\limits_{m \in \mathcal{M}} \mathrm{N_{max}}(m)$ for fair comparison with the proposed scheme. In addition, the transmit power of all end-devices are set equal to $P_\mathrm{max}=14$ [dBm], since no power allocation schemes had been proposed to jointly tackle co-SF and inter-SF interferences so far.
%\begin{enumerate}
\begin{itemize}
    \item \textbf{Random SF-allocation (\emph{Conv. Random})}: the gateway chooses randomly $A$ devices among $\mathcal{N}$ and assigns a random SF to each of these devices among the possible SFs. Details of this scheme are given in Algorithm~\ref{algo:random}.

\begin{algorithm}[h]
\small
\caption{Baseline scheme: Random SF-Allocation}
\label{algo:random}
\textbf{Initialization: } Set of unmatched end-devices: $\mathcal{L}_{U}\gets \mathcal{N}$, list of end-devices assigned to $\mathrm{SF}_m$: $\mathcal{A}_m \gets \emptyset$.
\begin{algorithmic}[1]
\State i$\gets$0;\Comment{Choose $A=\sum\limits_{m \in \mathcal{M}}\mathrm{N_{max}}(m)$ random end-devices and assign them to a random $\mathrm{SF}_m$}
\While{$i < A$}
\State $d \gets \mathit{random}(\mathcal{L}_{U})$; 
\State $\mathcal{L}_{U} \gets \mathcal{L}_{U} \backslash \{d\}$;
\State $j$ $\gets \mathit{random}(\mathcal{M})$;
\State $\mathcal{A}_j.\mathit{add}(d)$;
\State $i$ $\gets$ $i+1$;
\EndWhile
\end{algorithmic}
\end{algorithm}

\item \textbf{Distance SF-Allocation (\emph{Conv. Distance}}): the gateway chooses randomly $A$ devices among $\mathcal{N}$. Then, the SF for each of these devices is determined by Table~\ref{tab:tab1} based on their distance $r_n$: device $n$ uses $\mathrm{SF}_m$ if $r_n \in (l_{m-1},l_m]$. Details of this scheme are given in Algorithm~\ref{algo:distance}.

\begin{algorithm}[ht]
\small
\caption{Baseline scheme: Distance SF-Allocation}
\label{algo:distance}
\textbf{Initialization:} Set of unmatched end-devices: $\mathcal{L}_{U}\gets \mathcal{N}$, list of end-devices assigned to $\mathrm{SF}_m$: $\mathcal{A}_m\gets \emptyset$ .
\begin{algorithmic}[1]
\State $i\gets 0$;
\Comment{Choose $A=\sum\limits_{m \in \mathcal{M}}\mathrm{N_{max}}(m)$ random end-devices and assign them to the SF}
\While{$i < A$}
\State $d \gets \mathit{random}(\mathcal{L}_U)$;
\State $\mathcal{L}_U \gets \mathcal{L}_U \backslash \{d\}$;
\For{$j \in \mathcal{M}$}
\If{$l_j\geq d.dist \And l_{j-1}\leq d.dist$}
\State $\mathcal{A}_j.\mathit{add}(d)$;
\EndIf
\EndFor
\State $i \gets i+1$;
\EndWhile
\end{algorithmic}
\end{algorithm}
\vspace{-.5cm}
\end{itemize}
\subsection{Choice of $N_{max}$ given a target minimum throughput}\label{choice_nmax}
To determine the quota of each SF, we fix a target minimal throughput equal to 1 bit/s. We have run preliminary simulations over 100000 frames. Table~\ref{tab:tab3} represents the minimal short-term average rate achieved on each SF, for different values of $\mathrm{N_{max}}$. We can observe that to guarantee the target minimal throughput of 1 bit/s, we can have at most three devices assigned to $\mathrm{SF}_7$ but only one  device to the other SFs. In the sequel, we consider two scenarios: firstly, where there is no co-SF interferences, i.e., $\mathrm{N_{max}}(m)=1 \text{ }\forall m$, and secondly, where both co-SF and inter-SF interferences are present, with at most three end-devices assigned to $\mathrm{SF}_{7} \left(\mathrm{N_{max}}(7)=3\right)$ and one to the others $\left(\mathrm{N_{max}}(m)=1 \text{ }\forall m\neq 7\right)$. For fair comparison, the number of simultaneously transmitting end-devices $A$ is equal to $\sum\limits_{m\in\mathcal{M}} \mathrm{N_{max}}(m)$ in each allocation period.

%\newcolumntype{C}[1]{>{\centering\arra%ybackslash}p{#1}}
%\begin{table}[t]
%\resizebox{\textwidth}{!}{
%\begin{tabular}{|C{1.6cm}|C{1.86cm}|C{%1.86cm}|C{1.86cm}|C{1.86cm}|C{1.86cm}|%C{1.86cm}|}
%\hline
%        $\mathrm{\mathbf{N_{max}}}$&$\%mathrm{\mathbf{SF}_7}$&$\mathr%m{\mathbf{SF}_8}$&$\mathrm{\ma%thbf{SF}_9}$&$\mathrm{\mathbf{%SF}_{10}}$&$\mathrm{\mathbf{SF%}_{11}}$&$\mathrm{\mathbf{SF}_%{12}}$\\ 
%        \hline
%      1&
%\cellcolor[HTML]{E2FAFF}{4.82}&\cellco%lor[HTML]{E2FAFF}{1.51} %&\cellcolor[HTML]{E2FAFF}{1.06} & %\cellcolor[HTML]{E2FAFF}{4.7\mathrm{e}%-1}& \cellcolor[HTML]{E2FAFF}{2.7\math%rm{e}-1}&\cellcolor[HTML]{E2FAFF}{1.9\%mathrm{e}-1}\\ 
%     \hline
%      2& \cellcolor[HTML]{E2FAFF}{7.7e-2}&1.1\mathrm{e}-7 &9.3\mathrm{e}-14& 7.8\mathrm{e}-25& 6.7e-46&3.7e-78\\
%     \hline
%     3& \cellcolor[HTML]{E2FAFF}{2.7e-3}&8.2e-9 &2e-15 &8.2e-27 & 3.1\mathrm{e}-49&1.3\mathrm{e}-84 \\
%     \hline
%       4& 9.9\mathrm{e}-5&5.8\mathrm{e}-10 & 9.0\mathrm{e}-17& 4.3\mathrm{e}-29&1.2\mathrm{e}-49 &1.1\mathrm{e}-86 \\
%       \hline
%     5& 1.8\mathrm{e}-6&5.2\mathrm{e}-11 &6.5\mathrm{e}-18 & 1.3\mathrm{e}-30&1.0\mathrm{e}-53 & 3.7\mathrm{e}-93\\
%     \hline
%\end{tabular}
%}
%\vspace{.1cm}
%\caption{Minimal throughput for each $\mathrm{SF}_m$ (in kbit/s)}
%\label{tab:tab3} 
%\vspace{-1.5cm}
%\end{table}

%\small
%\newcolumntype{C}[1]{>{\centering\arraybackslash}p{#1}}
\begin{table}[ht]
\resizebox{\textwidth}{!}{
\begin{tabular}{|C{1.6cm}|C{1.86cm}|C{1.86cm}|C{1.86cm}|C{1.86cm}|C{1.86cm}|C{1.86cm}|}
\hline
        $\mathrm{\mathbf{N_{max}}}$&$\mathrm{\mathbf{SF}_7}$&$\mathrm{\mathbf{SF}_8}$&$\mathrm{\mathbf{SF}_9}$&$\mathrm{\mathbf{SF}_{10}}$&$\mathrm{\mathbf{SF}_{11}}$&$\mathrm{\mathbf{SF}_{12}}$\\ 
        \hline
        \cline{2-7}
      1&
\multicolumn{1}{!{\vrule width 1pt}c!{\vrule  width .5pt}}{4.82}&{1.51} &{1.06} & {4.7$\mathrm{e}$-1}& {2.7$\mathrm{e}$-1}& \multicolumn{1}{!{\vrule}c!{\vrule width 1pt}}{1.9$\mathrm{e}$-1}\\
     \hline
     \cline{3-7} 
      2& \multicolumn{1}{!{\vrule width 1pt}c!{\vrule width 1pt}}{7.7e-2}&1.1$\mathrm{e}$-7 &9.3$\mathrm{e}$-14& 7.8$\mathrm{e}$-25& 6.7e-46&3.7e-78\\
     \hline
          %\cline{2-1} 
     3& \multicolumn{1}{!{\vrule width 1pt}c!{\vrule width 1pt}}{2.7e-3}&8.2$\mathrm{e}$-9 &2$\mathrm{e}$-15 &8.2$\mathrm{e}$-27 & 3.1$\mathrm{e}$-49&1.3$\mathrm{e}$-84 \\
     \hline
     \cline{2-1} 
      4& 9.9$\mathrm{e}$-5&5.8$\mathrm{e}$-10 & 9.0$\mathrm{e}$-17& 4.3$\mathrm{e}$-29&1.2$\mathrm{e}$-49 &1.1$\mathrm{e}$-86 \\
     \hline
     5& 1.8$\mathrm{e}$-6&5.2$\mathrm{e}$-11 &6.5$\mathrm{e}$-18 & 1.3$\mathrm{e}$-30&1.0$\mathrm{e}$-53 & 3.7$\mathrm{e}$-93\\
     \hline
\end{tabular}
\vspace{.1cm}
}
\caption{Minimal throughput for each $\mathrm{SF}_m$ (in kbit/s)}
\label{tab:tab3} 
\vspace{-1.5cm}
\end{table}

\subsection{Performance Evaluation for $N_{max}(m)=1, \forall m \in \mathcal{M}$}
First let us discuss the case with SF allocation optimization only with maximum transmit power as in baseline schemes, namely the performances of \emph{Prop. Initial} (Algorithm~\ref{algo:Initial_matching}) and \emph{Prop. SF allocation} (Algorithms~\ref{algo:Initial_matching} and~\ref{algo:Matching_refinement}). Figure~\ref{fig:fj_min_1-1-1-1-1-1} shows the performance comparison of our proposed algorithms, with and without the power optimization step, and the baseline schemes in terms of minimal short-term average rates as a function of a varying number of end-devices. We can observe that our proposed algorithm yields significant performance gains compared to both the random SF-allocation and distance SF-allocation for all values of $N$. For instance, Figure~\ref{fig:fj_min_1-1-1-1-1-1} shows that, while baseline schemes lead to an early drop of minimal rate (almost null for $N>6$), the proposed algorithms still provide a good minimal throughput for a much higher number of end-devices. In this case, \emph{Prop. SF allocation} and  \emph{Prop. Initial} perform similarly.
\begin{figure}[H]
    \vspace{-.5cm}
    \centering
    \includegraphics[width=9.5cm]{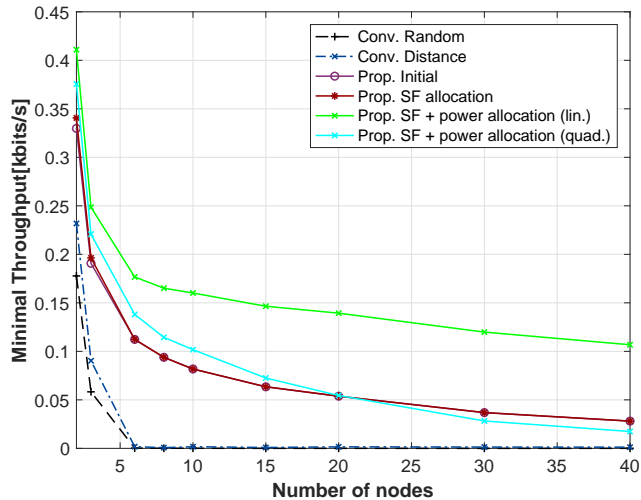}
    \vspace{-.5cm}
    \caption{Minimal short-term average rates, for proposed and baseline algorithms}
    \label{fig:fj_min_1-1-1-1-1-1}
    \vspace{-.5cm}
\end{figure}

\begin{figure}[H]
    \vspace{-.5cm}
    \centering
    \includegraphics[width=9.5cm]{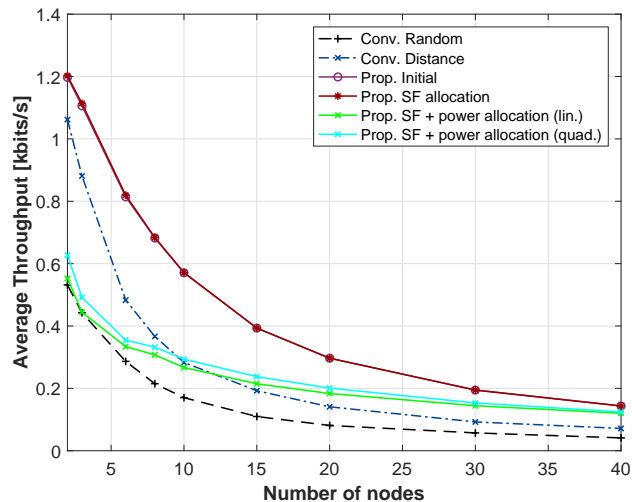}
    \vspace{-.5cm}
    \caption{Average network throughput for proposed and baseline algorithms}
    \label{fig:fj_avg_1-1-1-1-1-1}
    \vspace{-.5cm}
\end{figure}

Figure~\ref{fig:fj_avg_1-1-1-1-1-1} shows the performance comparison in terms of average network throughput over all end-devices between the different allocation schemes, against a varying number of end-devices. We can clearly see that the proposed scheme is superior to all other schemes. From Figure~\ref{fig:fj_avg_1-1-1-1-1-1}, the proposed method can provide an average throughput always larger than 180 bit/s while \emph{Conv. Random} and \emph{Conv. Distance} offer less than half for $N \geq 10$. We can also notice that \emph{Conv. Distance} performs quite good when $N \leq 10$.

We now evaluate the fairness levels of the different algorithms by using the Jain's fairness index, given by $ \mathcal{J}=\frac{\left(\sum\limits_{n\in \mathcal{N}} U_n\right)^2}{N \times \sum\limits_{n\in \mathcal{N}} U_n^2}$. 
\begin{figure}[H]
    \vspace{-.5cm}
    \centering
    \includegraphics[width=9.5cm]{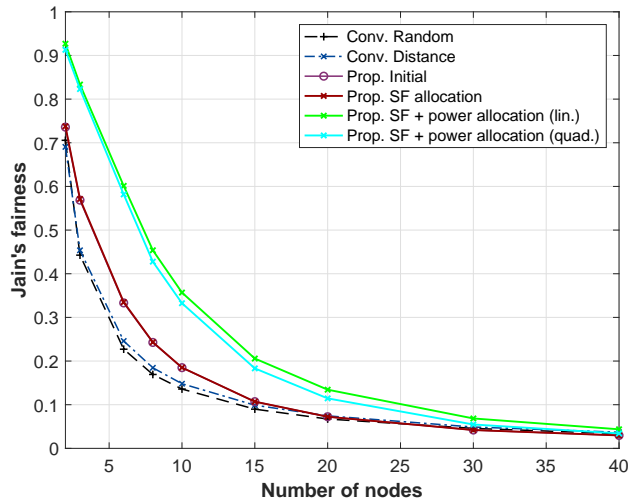}
    \vspace{-.5cm}
    \caption{Jain's fairness metric for proposed and baseline algorithms}
    \label{fig:fj_tp_Jf_1-1-1-1-1-1}
\end{figure}
\vspace{-.5cm}

Figure~\ref{fig:fj_tp_Jf_1-1-1-1-1-1} shows that by the considered max-min strategy and matching-based methodology, the proposed algorithms improve well the system fairness level compared to baseline methods.

Next, we discuss the performance of the proposed joint SF and power allocation algorithms, shown in Figs.~\ref{fig:fj_min_1-1-1-1-1-1},~\ref{fig:fj_avg_1-1-1-1-1-1}, and~\ref{fig:fj_tp_Jf_1-1-1-1-1-1}
%In the previous Figures~\ref{fig:fj_min_1-1-1-1-1-1},~\ref{fig:fj_avg_1-1-1-1-1-1} and~\ref{fig:fj_tp_Jf_1-1-1-1-1-1} and Figure~\ref{fig:fj_tp_Pow_1-1-1-1-1-1}, we can see the performance achieved by the algorithms of the joint SF and transmit power optimization in addition to the SF-allocation algorithm and the baseline schemes. We can draw the following conclusion from these figures:
. Firstly, we observe that the proposed power allocation in Algorithm~\ref{algo:power_1}, with linear and quadratic approximations, outperforms all other methods, including the proposed SF-allocation with fixed power, in terms of minimal short-term average rates. However, when $N$ is larger than 20, we can observe a decrease in the minimal throughput of the quadratic approximation compared to our proposed SF-allocation algorithm. This is due to the use of the quadratic approximation which does not necessarily guarantee a better local optimum compared to that offered by linear approximation, as this depends on the difference between the solution sets of the approximated problems - linear and quadratic cases -, and that of the original problem. However, both approximations yield much higher minimal throughputs compared to baseline schemes. 
Along with higher minimal throughput, Fig.~\ref{fig:fj_tp_Jf_1-1-1-1-1-1} shows the large fairness improvements brought by our joint SF and power allocation schemes, against baseline and proposed scheme with SF-allocation only.
Furthermore, with Fig.~\ref{fig:fj_avg_1-1-1-1-1-1}, we observe that with optimized power, the proposed solutions enable much larger minimal throughput and higher Jain's fairness, but at the cost of lower network throughput. Still, the proposed schemes, with both linear and quadratic approximations, outperform both baseline schemes in terms of network throughput, for larger number of end-devices.

\begin{figure}[ht]
    \vspace{-.5cm}
    \centering
    \includegraphics[width=9.5cm]{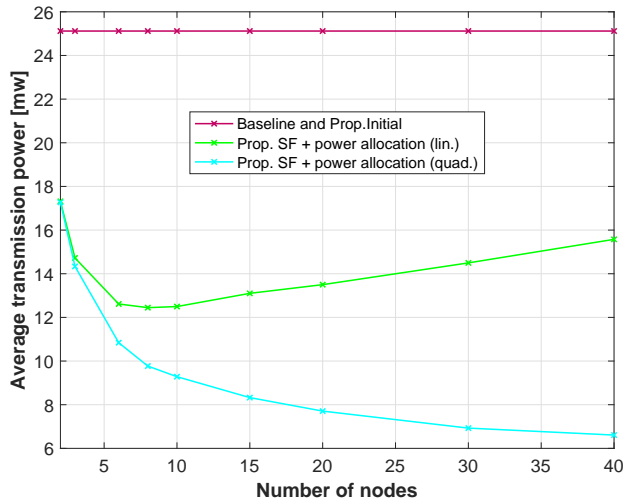}
    \vspace{-.5cm}
    \caption{Average power consumption for proposed and baseline algorithms}
    \label{fig:fj_tp_Pow_1-1-1-1-1-1}
    \vspace{-.5cm}
\end{figure}

Finally, Fig.~\ref{fig:fj_tp_Pow_1-1-1-1-1-1} depicts the average transmit power consumed by end-devices with a varying number of nodes. We can observe that the proposed joint allocation schemes enable important savings in energy consumption while providing better throughput and higher fairness compared to the fixed transmit power allocation approaches. We also notice that with a quadratic approximation, Algorithm~\ref{algo:power_1} allows even higher power savings, i.e., up to 58\% compared to the linear approximation case. That is, in the linear case, more power is spent for low channel quality users in order to maintain high minimal average rates. On the contrary, solutions obtained by quadratic approximation tend to decrease power consumption, at the expense of lower minimal throughputs.

\subsection{Performance Evaluation for $\mathrm{N_{max}}(m)= 3$ for $m=7$, $\mathrm{N_{max}}(m)=1 \text{ } \forall m \neq 7$}
For the second scenario, Fig.~\ref{fig:fj_min_3-1-1-1-1-1} depicts the performance comparison of our proposed algorithms with and without joint power allocation optimization step, and the baseline schemes. From Fig.~\ref{fig:fj_min_3-1-1-1-1-1} we can first confirm that our SF-allocation algorithm \emph{Prop. SF allocation} still outperforms baseline schemes even when increasing $\mathrm{N_{max}}(7)$. However, its performance decreases compared to the case of Fig.~\ref{fig:fj_min_1-1-1-1-1-1} where there are no co-SF interferences. We also observe that, unlike in the previous scenario, \emph{Prop. SF allocation} now provides higher minimal throughputs than \emph{Prop. Initial} for $N\leq 20$: this performance gap is more obvious than the case where $\mathrm{N_{max}}(7)=1\text{ } \forall m,$ since swap operations were almost absent in that case.
\begin{figure}[ht]
    \vspace{-.5cm}
    \centering
    \includegraphics[width=9.5cm]{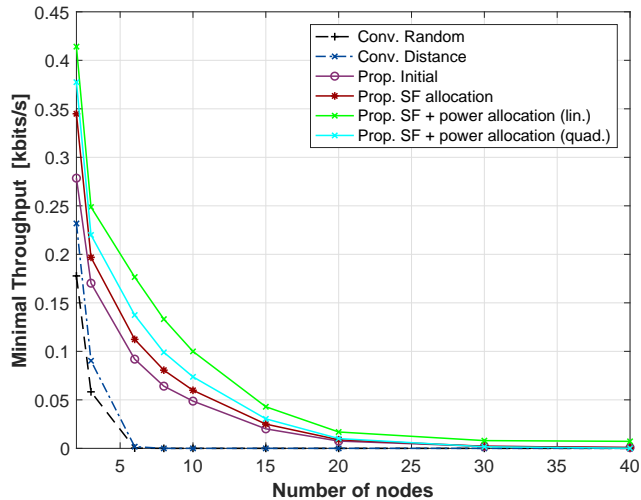}
    \vspace{-.5cm}
    \caption{Minimal short-term average rates for proposed and baseline algorithms}
    \label{fig:fj_min_3-1-1-1-1-1}
    \vspace{-.5cm}
\end{figure}

\begin{figure}[H]
    \vspace{-.5cm}
    \centering
    \includegraphics[width=9.5cm]{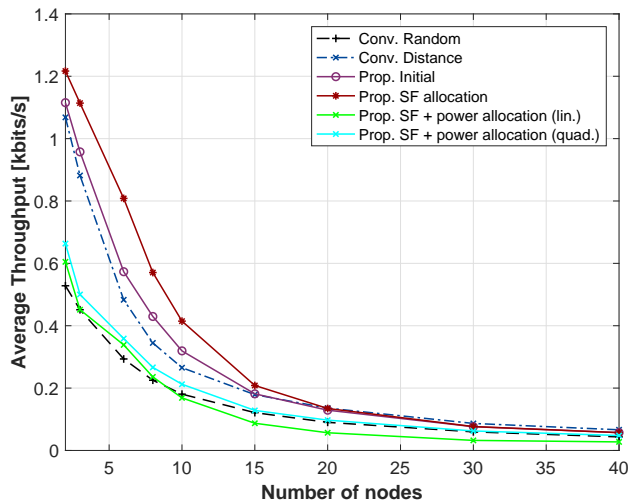}
    \vspace{-.5cm}
    \caption{Average network throughput for proposed and baseline algorithms}
    \label{fig:fj_avg_3-1-1-1-1-1}
    \vspace{-.5cm}
\end{figure}

Fig.~\ref{fig:fj_avg_3-1-1-1-1-1} shows the impact of maximizing the minimal short-term average rates on the average network throughput over all end-devices. It can be clearly seen that the highest average network throughput is achieved by our \emph{Prop. SF allocation} for $N \leq 20$ and that it provides a significant improvement compared to Initial Matching, i.e., \emph{Prop. Initial}. However, with joint power allocation, the proposed schemes have a poorer performance as maintaining a high minimal throughput is very challenging whenever there are both inter-SF and co-SF interferences. Note that the proposed solution with quadratic approximation offers a slightly higher average throughput compared to the linear case.

\begin{figure}[H]
    \vspace{-.5cm}
    \centering
    \includegraphics[width=9.5cm]{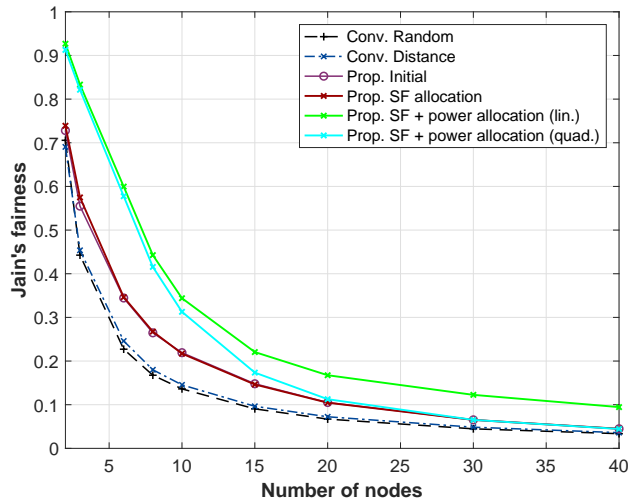}
    \vspace{-0.5cm}
    \caption{Jain's fairness metric,  proposed and baseline algorithms}
    \label{fig:fj_tp_jf_3-1-1-1-1-1}
    \vspace{-.5cm}
\end{figure}

From Figure~\ref{fig:fj_tp_jf_3-1-1-1-1-1}, we clearly see that the proposed approaches bring significant performance gains in terms of fairness, which is in line with the gains achieved in terms of minimal throughputs. In addition, the proposed power optimization still enables remarkable fairness improvements, even under both inter-SF and co-SF interferences, with a larger gain for the linear approximation.

\begin{figure}[H]
    \vspace{-.5cm}
    \centering
    \includegraphics[width=9.5cm]{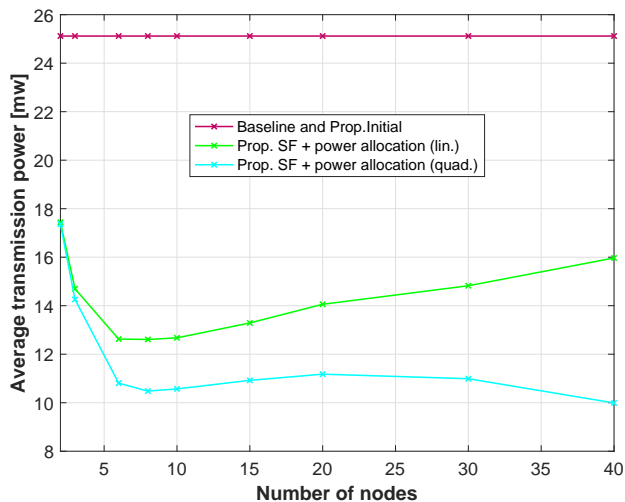}
    \vspace{-.5cm}
    \caption{Average power consumption for proposed, proposed with power control and baseline algorithms}
    \label{fig:fj_tp_pow_3-1-1-1-1-1}
    \vspace{-.5cm}
\end{figure}
Finally, Fig.~\ref{fig:fj_tp_pow_3-1-1-1-1-1} shows the drastic energy savings offered by our proposed power optimization strategies (up to 37.43\%), compared to baseline and proposed SF-allocation only. Similarly to scenario 1, the quadratic approximation offers further power savings compared to the linear approximation case.

Overall, the proposed joint SF assignment and power allocation method provides remarkable performance improvements, jointly in terms of minimal achievable rates, network throughput, fairness, and consumed energy with a limited computational complexity suitable for LoRa gateways.

\section{Conclusion}
\label{sec:conclude}
In this work, we have addressed the issue of network performance enhancement for a LPWAN based on LoRa physical layer, where both impacts of co-SF and inter-SF interferences were included. Focusing on user fairness improvement for uplink communications, the objective was to optimize SFs' assignment and transmit power allocation for maximizing the minimal short-term average user rates, whose expressions are in line with the LoRaWAN specifications by not assuming instantaneous CSIs. The intractability of the joint SF and power allocation problem is tackled by separating it into two subproblems: SF assignment under fixed power, and power allocation under fixed SFs. Simulation results show that, despite severe co-SF and inter-SF interferences, our proposed algorithms outperformed baseline algorithms, jointly in terms of minimal rates, user fairness, and network throughput. Both proposed linear and quadratic approximation approaches to the non-linear feasibility problem for power allocation were shown to provide efficient transmit power solutions, leading to drastic energy savings, while further enhancing minimal throughput and user fairness.

% trigger a \newpage just before the given reference
% number - used to balance the columns on the last page
% adjust value as needed - may need to be readjusted if
% the document is modified later
%\IEEEtriggeratref{8}
% The "triggered" command can be changed if desired:

%\IEEEtriggercmd{\enlargethispage{-5in}}

%\appendix
%\input{append2}

%\input{append_sampling}
% that's all folks
\section*{Acknowledgments}
\label{sec:ack}
This work was supported by the NII Collaborative Research Grant, the NII Grant for the MoU with LIMOS University Clermont Auvergne, and by the Grant-in-Aid for Scientific Research (Kakenhi) no. 17K06453 from the Ministry of Education, Science, Sports, and Culture of Japan.

\bibliographystyle{IEEEtran}
\bibliography{ref2}

%\newpage

%\appendix
%\label{sec:appen}
%\input{appenA}
%\input{appenB}

\end{document}